\documentclass[onecolumn]{IEEEtran}

\usepackage{amsfonts}
\usepackage{algorithmic}
\usepackage{algorithm}
\usepackage{amssymb}
\usepackage[cmex10]{amsmath}
\usepackage{amsthm,mathtools}
\usepackage{epsf} 
\usepackage{psfrag}
\usepackage{times,color}
\usepackage{cases}
\usepackage{graphicx}
\usepackage{float}
\usepackage{color}
\usepackage{graphicx}
\usepackage{times,color}
\usepackage{amsfonts}
\usepackage{bm}
\usepackage[draft,ulem=normalem]{changes}
\newtheorem{theorem}{Theorem}

\newtheorem{lemma}{Lemma}

\newtheorem{property}{Property}
\newtheorem{definition}{Definition}

\newtheorem{claim}{Claim}
\usepackage{mleftright}\mleftright
\newcommand{\nchoosek}[2]{\left(\begin{array}{c}#1\\#2\end{array}\right)}
\newcommand{\comment}[1]{}

\usepackage{cleveref}
\Crefname{claim}{Claim}{Claims}
\Crefname{lemma}{Lemma}{Lemmas}
\Crefname{theorem}{Theorem}{Theorems}
\Crefname{property}{Property}{Properties}
\Crefname{definition}{Definition}{Definitions}
\newtheorem{exam}{Example$\!$}
\newenvironment{example}{\begin{exam}\hspace*{-1ex}{\bf }}{\end{exam}}
\newenvironment{proof-sketch}{\noindent \textit{Sketch of Proof:}}{$\blacksquare$}

\newtheorem{remrk}{Remark$\!$}

\newcommand{\cB}{{\cal B}}
\newcommand{\cC}{{\cal C}}
\newcommand{\cD}{{\cal D}}

\newcommand{\cF}{{\cal F}}
\newcommand{\cG}{{\cal G}}

\newcommand{\cS}{{\cal S}}

\newcommand{\cX}{{\cal X}}

\newcommand{\bfa}{{\boldsymbol a}}
\newcommand{\bfb}{{\boldsymbol b}}
\newcommand{\bfc}{{\boldsymbol c}}

\newcommand{\bfe}{{\boldsymbol e}}

\newcommand{\bfs}{{\boldsymbol s}}

\newcommand{\bfu}{{\boldsymbol u}}
\newcommand{\bfv}{{\boldsymbol v}}
\newcommand{\bfw}{{\boldsymbol w}}
\newcommand{\bfx}{{\boldsymbol x}}
\newcommand{\bfy}{{\boldsymbol y}}
\newcommand{\bfz}{{\boldsymbol z}}
\newcommand{\bfsigma}{{\boldsymbol{\sigma}}}

\newcommand{\wt}{\operatorname{wt}}

\begin{document}

\title{Reconciling Similar Sets of Data}

\author{
  \IEEEauthorblockN{
    Ryan~Gabrys\IEEEauthorrefmark{1},
    Farzad~Farnoud~(Hassanzadeh)\IEEEauthorrefmark{2}}\thanks{This paper was presented in part at the International Symposium on Information Theory 2015 in Hong Kong and at the 55th Annual Allerton Conference on Communication, Control, and Computing in 2017. \cite{conf1}, \cite{conf2}}\\
  {\normalsize
    \begin{tabular}{cc}
      \IEEEauthorrefmark{1}Spawar Systems Center San Diego~~ & \IEEEauthorrefmark{2}University of Virginia\\
           %La Jolla, CA 92093, U.S.A.~~~~~~ & La Jolla, CA 92093, U.S.A. \\
           ryan.gabrys@navy.mil & farzad@virginia.edu
    \end{tabular}}
    }

\maketitle

\begin{abstract} In this work, we consider the problem of synchronizing two sets of data where the size of the symmetric difference between the sets is small and, in addition, the elements in the symmetric difference are related through the Hamming distance metric. Upper and lower bounds are derived on the minimum amount of information exchange. Furthermore, explicit encoding and decoding algorithms are provided for many cases.
\end{abstract}

\section{Introduction}\label{sec:intro}

Suppose two hosts, A and B, each have a set of length-$n$ $q$-ary strings. Let $\cS^A$ denote the set of strings on Host $A$ and let $\cS^B$ denote the set of strings on Host $B$. The \textbf{\textit{set reconciliation}} problem is to determine the minimum information that must be sent from Host $A$ to Host $B$ with a single round of communication so that Host $B$ can compute their symmetric difference $\cS^A \bigtriangleup \cS^B= (\cS^A \setminus \cS^B) \cup (\cS^B \setminus \cS^A)$ where $|\cS^A \bigtriangleup \cS^B| \leq t$.

This problem has been the subject of study in many works such as \cite{MT02},  \cite{KLT03}, \cite{MTZ03}, \cite{EGUV11}, \cite{GL13}, and \cite{SR14}. The work in \cite{MT02} provides an approach to set reconciliation using polynomial interpolation. In \cite{KLT03} and \cite{MTZ03}, coding schemes were studied that were based upon error-correcting codes and polynomial interpolation. In \cite{EGUV11} and \cite{GL13} algorithms for set reconciliation were considered based upon Bloom filters. In \cite{SR14}, the authors consider the problem of synchronizing vector subspaces.

%In \cite{EGUV11} and \cite{GL13} algorithms for set reconciliation were considered based upon Bloom filters. 

 In this paper, we consider a variant of the traditional set reconciliation problem whereby the elements in the symmetric difference $\cS^A \bigtriangleup \cS^B$ are related. This setup could arise, for instance, when users are synchronizing files that are being edited or when the data elements themselves are interrelated. This paper focuses on the generic setup where the symmetric difference can be partitioned into subsets such that elements in each of these subsets are within a certain distance of each other. The focus in this work will be on transmission schemes that minimize the amount of information exchanged between two hosts.
 
%\added{In \cite{CKYYZ14}, the authors consider a problem which takes into account the similarities between objects; however, the problem studied is to minimize the earth movers distance between two sets given an upper bound on the communication cost.}
 
\begin{figure}
\begin{center}
\includegraphics[width=1\columnwidth]{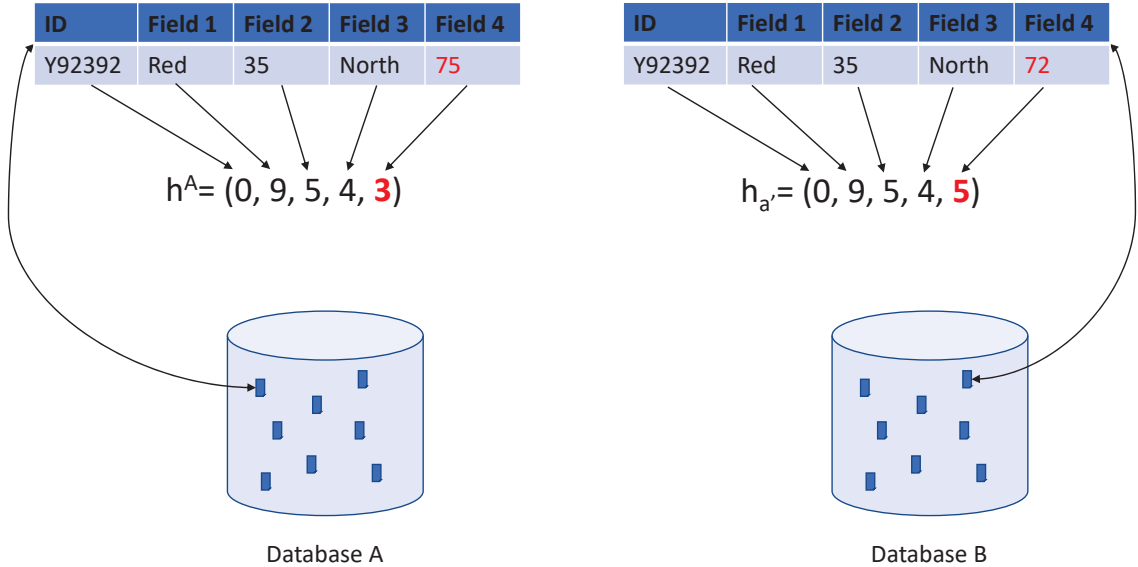}
\caption{Two Databases Synchronizing Related Sets of Information.}
\label{fig:dbexample}
\end{center}
\end{figure}

Specifically, the model studied in this work is motivated by the scenario where two hosts are storing a large number of (potentially large) documents. Under this setup, information is never deleted so that each database contains many different versions of the same document. Each document has a fixed number of fields and each field has a fixed size. When synchronizing sets of documents between two hosts, a set of hashes is produced. For every document, a single hash is formed by concatenating (in a systematic fashion) the result of hashing each field of the document. 

In Figure~\ref{fig:dbexample}, this setup is illustrated with the rectangles representing documents stored within a database. Each document consists of a unique ID along with $4$ additional name-value pairs. On the lefthand side of the diagram, we  show a document with the ID $Y92392$, which we will refer to as document $a$ for shorthand. Suppose $h^A = (0,9,5,4,3) \in \mathbb{Z}_{10}^5$ is the result of performing the hash described in the previous paragraph on document $a$. Suppose a single field on document $a$ is updated resulting in the document $a'$ (which is shown in Figure~\ref{fig:dbexample} as residing on Host $B$) and that $h_{a'}=(0,9,5,4,\textbf{5}) \in \mathbb{Z}_{10}^5$ is the hash for ${a'}$. By the previous discussion, $h^A$ and $h_{a'}$ differ only in the portion of $h_{a'}$ which corresponds to the field that was updated and the Hamming distance between $h^A$ and $h_{a'}$ is one. Motivated by this setup, we study the problem of reconciling sets of elements whereby subsets of elements in the symmetric difference are within a bounded Hamming distance from each other.

The contributions of this work include bounds and coding schemes for reconciling sets of related strings that reduce the information exchange. As will be discussed in more details, we derive transmission schemes that, under certain conditions, require less information exchange than existing, alternative methods. 

The paper is organized as follows. In \Cref{sec:model}, we formally define our problem and introduce some useful notation. Upper and lower bounds on the amount of required information exchange are provided in \Cref{sec:bounds}. In \Cref{sec:r1sets}, we provide a coding scheme for reconciling certain sets of related information.  In Section~\ref{sec:t>1}, we consider an extension of the ideas from \Cref{sec:r1sets} that can be used for reconciling more generic sets of related information. Section~\ref{sec:conclude} concludes the paper.
% are interested in designing codes for the setup where the elements in $\cS^A \bigtriangleup \cS^B$ are related. We first introduce some terminology.

\section{Model and Preliminaries}\label{sec:model}

For two strings $\bfx, \bfy \in \mathbb{F}_q^n$, let $d_H(\bfx, \bfy)$ denote their Hamming distance. We denote the Hamming weight of $\bfx$ as $\wt(\bfx)$. Throughout this paper, we assume $q$ is a power of $2$ and a constant.

\begin{definition}\label{def:thlset} Let $\cS^A \subseteq \mathbb{F}_q^n$ and $\cS^B \subseteq \mathbb{F}_q^n$. We say that $(\cS^A, \cS^B)$ are $(t, h, \ell)$-sets if $\cS^A \bigtriangleup \cS^B$ can be be written
\begin{align*}
 \cS^A \bigtriangleup \cS^B &=\bigcup_{i=1}^{j}\{ \bfx_{i,1}, \ldots, \bfx_{i,k_i} \}, 
\end{align*}
where
\begin{enumerate}
%\item for $1 \leq i \leq j$, $1 \leq k_i \leq h$,
\item $j \leq t$, 
\item for $1 \leq i \leq j$, $k_i \leq h$, and
%\item $|\cS^A \bigtriangleup \cS^B| = 2m \leq t$, and
\item for any $\bfu, \bfw \in \{\bfx_{i,1}, \ldots \bfx_{i,k_i} \}$, we have $d_H(\bfu, \bfw) \leq \ell$.
\end{enumerate}
\end{definition}
An illustration is given in Figure~\ref{fig:set-recon}, where $t=3$ and an example is provided next.
  
\begin{example} {Suppose $\cS^A, \cS^B \in \mathbb{F}_2^5$, where 
\begin{equation*}
\begin{split}
\cS^A&=\{(0,0,0,0,0), (1,0,1,1,1)\},\\
\cS^B&=\{(0,0,0,0,0), (1,1,0,0,1)\}. 
\end{split}
\end{equation*}
Then we say that $(\cS^A, \cS^B)$ are $(1,2,3)$-sets since 
\[\cS^A \bigtriangleup \cS^B=\{ (1,\textcolor{red}{0},\textcolor{red}{1},\textcolor{red}{1},1),(1,\textcolor{red}{1},\textcolor{red}{0},\textcolor{red}{0},1) \},\] 
which can be decomposed into $1$ set of size $2$ whereby the Hamming distance between any two elements is at most $3$.} \end{example}

In the next section, we begin by deriving upper and lower bounds on the required information exchange to synchronize $(t,h,\ell)$-sets.
\begin{figure}
\begin{center}
%\psfrag{xxx}[cc][cc][.77][0]{$\cS_A\cap\cS_B$}
%\psfrag{yyy}[cc][cc][.77][0]{$\subset \cS_A\bigtriangleup\cS_B$}
%\psfrag{zzz}[cc][cc][.77][0]{$\subset \cS_A\bigtriangleup\cS_B$}
%\psfrag{www}[cc][cc][.6][0]{$\le \ell$}
%\psfrag{qqq}[cc][cc][.77][0]{$\subset\cS_A\bigtriangleup\cS_B$}
\includegraphics[width=.80\columnwidth]{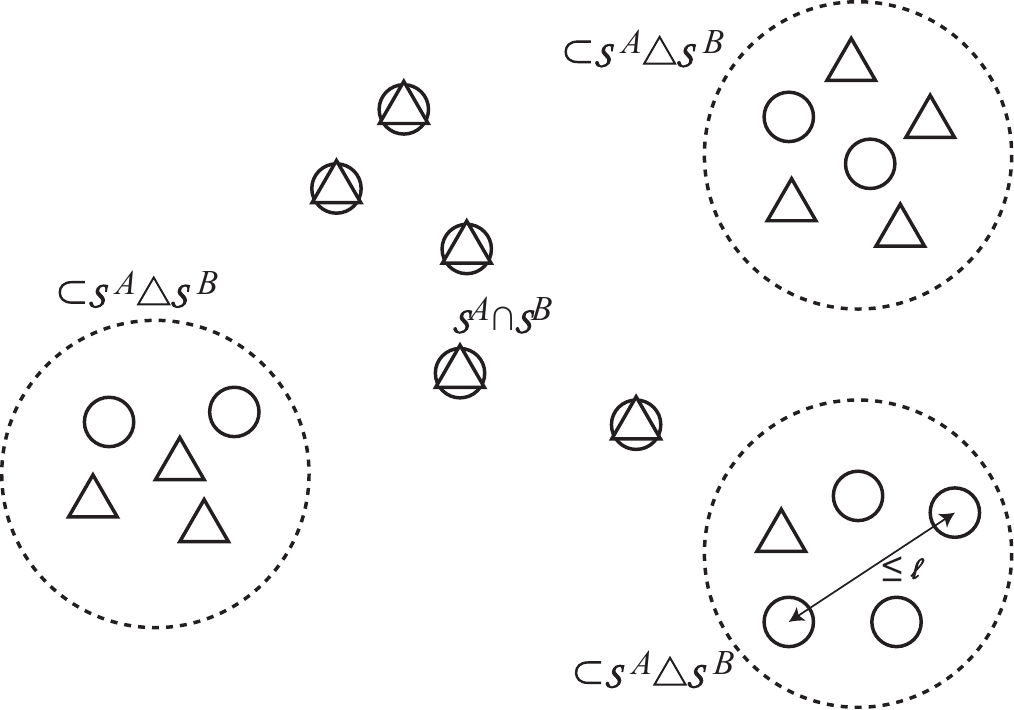}
\caption{An illustration for two $(t,h,\ell)$-sets $\cS^A$ (circles) and $\cS^B$ (triangles), where %subsets of 
$\cS^A\bigtriangleup\cS^B$ % are 
is divided into $t=3$ subsets each with at most $h=6$ elements such that any pair of elements within one of these subsets are at distance at most $\ell$.}
\label{fig:set-recon}
\end{center}
\end{figure}

\section{Bounds on Information Exchange}\label{sec:bounds}

We begin this section by introducing a graphical interpretation of our problem and then revisiting a result from~\cite{KLT03}. Afterwards, we consider non-asymptotic and asymptotic upper and lower bounds for the synchronization of $(t, h, \ell)$-sets.

Consider the undirected graph $\cG_{(t,h,\ell)}$ where each vertex in $\cG_{(t,h,\ell)}$ represents a set of length-$n$ $q$-ary strings. Notice that under this setup, $\cG_{(t,h,\ell)}$ has exactly $2^{q^n}$ vertices. There exists an edge between two vertices in $\cG_{(t,h,\ell)}$ if and only if the two vertices are $(t, h, \ell)$-sets. 

The following proposition from \cite{KLT03} depicts a close connection between the minimum information exchange for our set reconciliation problem and chromatic number of $\cG^2_{(t,h,\ell)}$. The square graph $\cG^2_{(t,h,\ell)}$ is a graph over the same vertex set as $\cG_{(t,h,\ell)}$ where two vertices are adjacent if their distance in $\cG_{(t,h,\ell)}$ is 1 or 2. 

\begin{theorem}\label{th:chrom}\emph{\cite[Theorem 1]{KLT03}}
The minimum information exchange required in one round of communication for reconciling any two $(t, h, \ell)$-sets is the chromatic number $\chi(\cG_{t,h,\ell}^2)$.
\end{theorem}

As a consequence of Theorem~\ref{th:chrom}, there exist hosts $A$ and $B$ where $(\cS^A, \cS^B)$ are $(t,h,\ell)$-sets such that at least $\log_2(\chi(\cG_{(t,h,\ell)}^2))$ bits of information is necessary for Host A to transmit to Host $B$ so that Host $B$ can determine $\cS^A \bigtriangleup \cS^B$. Theorem~\ref{th:bulbH} provides non-asymptotic upper and lower bounds for $\chi(\cG_{t,h,\ell}^2)$.

\begin{theorem}\label{th:bulbH} Let $\cG = \cG_{(t,h,\ell)}$ and suppose $\cC$ is a $q$-ary code with length $n$ and minimum Hamming distance $\ell+1$. We have
\begin{align*}
\sum_{j=1}^t \nchoosek{|\cC|}{j}& \left( \sum_{k=0}^{h-1} \nchoosek{r_1}{k} \right)^j \\
&\leq \chi(\cG^2)\leq 
\sum_{{j=0}}^{2t} \nchoosek{q^n}{j}  \left( \sum_{k=0}^{h-1} \nchoosek{r_2}{k}\right)^j,
\end{align*}
with 
$r_1 = \sum_{i=1}^{\lfloor \ell/2 \rfloor} \nchoosek{n}{i}  (q-1)^i,$
$r_2 = \sum_{i=1}^{\ell} \nchoosek{n}{i}(q-1)^i$.
%A = \frac{2^b}{\sum_{j=0}^{\ell} \nchoosek{b}{j}}$ and $r = \sum_{i=1}^{\lfloor \ell/2 \rfloor} \nchoosek{b}{i}.
\end{theorem}
\begin{proof}
In order to give a lower bound on $\chi(\cG^2)$, we give a lower bound on the size of the largest clique in $\cG^2$, which we denote as $\varsigma(\cG^2)$. The upper bound will be derived by providing an upper bound on the maximum degree of a vertex, which we denote as $\Delta(\cG^2)$. It is well-known (see \cite{W01} for instance) that 
$ \varsigma(\cG^2) \leq \chi(\cG^2) \leq \Delta(\cG^2) + 1. $

We produce a lower bound for $\varsigma(\cG^2)$ by considering the size of a clique in $\cG^2$. Suppose that Host $A$ contains no elements from $\mathbb{F}_q^n$ so that $\cS^A = \emptyset$. Let $\tilde{\cS}^B \subseteq \mathbb{F}_q^n$ be such that $|\tilde{\cS}^B| \leq t$ and for any $\bfx_1, \bfx_2 \in \tilde{\cS}^B$, we have $d_H(\bfx_1, \bfx_2) \geq \ell+1$. Let $\tilde{\cS}^B = \{ \bfx_1, \bfx_2, \ldots, \bfx_{t'} \}$ where $t' \leq t$. We form the set $\cS^B$, which represents the set of elements on Host $B$, from $\tilde{\cS}^B$ in the following manner. First we initialize $\cS^B = \tilde{\cS}^B$. Then, for $1 \leq i \leq t'$, add the elements of $\cS_i \subseteq \mathbb{F}_q^n$ to $\cS^B$, where $\cS_i$ is chosen such that for any $\bfy \in \cS_i$, $d_H(\bfx_i, \bfy) \leq \lfloor \frac{\ell}{2} \rfloor$ and $|\cS_i| \leq h-1$.

Under this setup, let $v_1$ be the vertex in $\cG$ representing the set $\cS^A$ and similarly let $v_2$ be the vertex in $\cG$ representing the set $\cS^B$. It is straightforward to observe that that since $(\cS^A, \cS^B)$ are $(t, h, \ell)$-sets, the distance between $v_1$ and $v_2$ in $\cG$ is one. Let $v_2, v_2'$ be two vertices in $\cG$ that represent two different possibilities for $\cS^B$. Then, by design the distance between $v_2$ and $v_2'$ is at most two so that the vertices are adjacent in $\cG^2$. Thus, the vertices $v_1, v_2, v_2'$ are all pairwise adjacent and they are part of a clique. Let $\cC \subseteq \mathbb{F}_q^n$ be a code with minimum Hamming distance $\ell+1$. From the previous discussion, notice that there are at least
$$\sum_{j=1}^t \nchoosek{|\cC|}{j} \cdot \left( \sum_{k=0}^{h-1} \nchoosek{r_1}{k} \right)^j $$
possible choices for the set $\cS^B$ where $r_1 = \sum_{i=1}^{\lfloor \ell/2 \rfloor} \nchoosek{n}{i} \cdot (q-1)^i$. 

%Applying the Gilbert-Varshamov bound \cite{R06} for $|\cC_{\ell+1}|$, the lower bound in the lemma follows.

We now produce an upper bound for $\Delta(\cG^2)$. Since the number of neighbors for a vertex $v \in \cG^2$ does not depend on the choice of $v$, we will simply assume $v \in \cG^2$ represents $\cS^A$ where, as before, $\cS^A=\emptyset$. Notice that if $v_1, v_2 \in \cG^2$ are adjacent, then $v_1, v_2$ represent $(2t, h, \ell)$-sets. We now count the number of possible choices for $\cS^B$ so that the sets $(\cS^A, \cS^B)$ are $(2t, h, \ell)$-sets under the assumption that $\cS^A = \emptyset$. We proceed similarly to before. Let $\tilde{\cS}^B \subseteq \mathbb{F}_q^n$ be any set of at most $2t$ elements from $\mathbb{F}_q^n$. We form the set $\cS^B$ in a manner similarly to before (two paragraphs up) given the set $\tilde{\cS}^B$. Under this setup, there are at most
$$ \sum_{j=1}^{2t} \nchoosek{q^n}{j} \cdot \left( \sum_{k=0}^{h-1} \nchoosek{r_2}{k} \right)^j  $$
possible choices for the set $\cS^B$ where $r_2 = \sum_{i=1}^{\ell} \nchoosek{n}{i}(q-1)^i$, which gives the upper bound in the lemma.
\end{proof}

\subsection*{Asymptotic Bounds}
\global\long\def\lb#1{\check{#1}}
\global\long\def\ub#1{\hat{#1}}

We now provide asymptotic upper and lower bounds for the information exchange. 

\begin{theorem}\label{th:asmptbnds}
	Let $\lambda=\frac{\ell}{n}$ and $\eta=\frac{\log_{q}h}{n}$. Assume
	$\lambda$ is between $0$ and $1-1/q$ and bounded away from both.
	If $n,h\to\infty$, 
	\begin{align*}
	\frac{\log_{q}\chi\left(\mathcal{G}^{2}\right)}{tnh} & \ge H_{q}\left(\frac{\lambda}{2}\right)-\eta+o(1),\\
	\frac{\log_{q}\chi\left(\mathcal{G}^{2}\right)}{tnh} & \le2(H_{q}\left(\lambda\right)-\eta)+o(1),
	\end{align*}
	where for the upper bound we also need $\eta<H_{q}\left(\lambda\right)-\epsilon$
	for some positive $\epsilon$.
\end{theorem}

The proof is given in the appendix. We compare these bounds\footnote{{In~\cite{conf1}, we had included an erroneous bound (eq.~(5)) that has been removed here. Furthermore, we have simplified the bounds compared to~\cite{conf1}.}} for fixed values of $\eta$ in \figurename~\ref{fig:bounds}. 
%Notice, that if we assume the elements in the symmetric difference are unrelated, then the minimum amount of required information exchange between two hosts is $t h \log_2(2^n) = t h n$ when binary data is being reconciled so that the rate of these schemes is necessarily greater than one (since we normalize by $thn$). Therefore, 

\begin{figure}
\begin{center}
\includegraphics[scale=0.85]{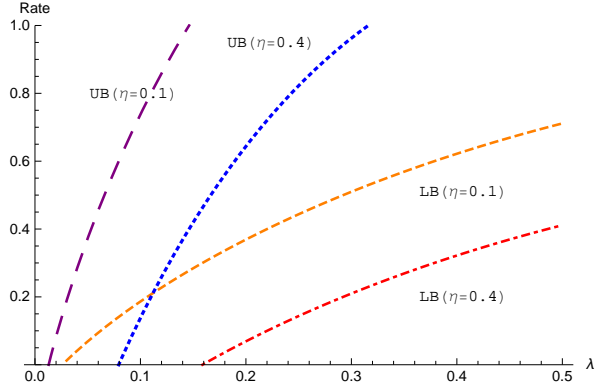}
\caption{Upper and lower bounds on the rate of information exchange, $\frac{\log_{q}\chi\left(\mathcal{G}^{2}\right)}{tnh}$, for $q=2$ and $n,h\to\infty$.}
\label{fig:bounds}
\end{center}
\end{figure}

Assume $q=2$. Then from Theorem~\ref{th:asmptbnds}, we have that the optimal number of bits of information exchange is at most $2thn(H(\ell/n)-\lg h /n)$. Notice that if the approach from \cite{MT02} was used, then at least $thn$ bits of information, which for $n,h$ large enough and small $\ell$, is more than $2thn(H(\ell/n)-\lg h /n)$. Thus, algorithms for reconciling $(t,h,\ell)$-sets have the potential to reduce the amount of information exchanged between hosts. As a starting point, in the next section, we consider an approach to reconciling $(1,h,\ell)$-sets.

\section{Reconciling $(1, h, \ell)$-sets}\label{sec:r1sets}

\newcommand{\comp}{C} % for renaming H_\comp 

In this section, we consider transmission schemes for the problem of reconciling $(1,h,\ell)$-sets, where $|\cS^A \bigtriangleup \cS^B|\le h$ and for all $\bfu, \bfw \in \cS^A \bigtriangleup \cS^B$, we have $d_H(\bfu, \bfw) \leq \ell$. We first describe the encoding procedure and then discuss the decoding method. Recall that the goal is to compute $\cS^A \bigtriangleup \cS^B$ where $(\cS^A, \cS^B)$ are $(1,h,\ell)$-sets consisting of elements from $\mathbb{F}_q^n$, where $q$ is a power of $2$. %Afterwards, we present an example illustrating the encoding and decoding 

The key idea behind the encoding and decoding is to encode the symmetric difference $\cS^A \bigtriangleup \cS^B$ by specifying one element say $X \in \cS^A \bigtriangleup \cS^B$ and then specifying the remaining elements in $\cS^A \bigtriangleup \cS^B$ by describing their location relative to $X$. As a result, as will be described shortly, the information transmitted from Host $A$ to Host $B$ can be decomposed into two parts denoted $\bfw_1$ and $\bfw_2$. The information in the $\bfw_1$ part describes the locations of the elements in $\cS^A \bigtriangleup \cS^B$ relative to $X$. The information in the $\bfw_2$ part will be used to fully recover $X$. Once $X$ is known and the locations of the other elements in $\cS^A \bigtriangleup \cS^B$ are known relative to $X$, then the symmetric difference $\cS^A \bigtriangleup \cS^B$ can be recovered.

We first introduce some useful notation. An $[n,d]_{q}$ code is a linear code over $\mathbb{F}_q$ of length $n$ with minimum Hamming distance $d$. Suppose $r$ is a positive integer where $r < n$. Let $\alpha$ be a primitive element in $\mathbb{F}_{q^r}$. Furthermore, let $H$ be an ${r \times n}$ parity check matrix with elements from  $\mathbb{F}_q$. Suppose $\cS = \{ \bfx_1, \bfx_2, \ldots, \bfx_s \} \subseteq \mathbb{F}_q^n$ and define {the syndrome of $\cS$ under $H$ as the multiset}
\[H \cdot \cS=\{ H \cdot \bfx_1, H \cdot \bfx_2, \ldots, H \cdot \bfx_s \}.\] Furthermore, define $\cS_{H,i}$ where $1 \leq i \leq q^r$ as \[\cS_{H,i} = \{ \bfx \in \cS : H \cdot \bfx = \alpha^i \},\] where with an abuse of notation, $\alpha^{q^r}=\bf0$. 
%Suppose $\cS_{O,H,i} =(\bfx_1, \ldots, \bfx_{|\cS_{H,i} |})$ is the lexicographic ordering of the elements in $\cS_{H,i}$. 
We refer to the $j$-th element in $\cS_{H,i}$, when ordered in lexicographic order, as $\cS_{H,i,j}$. Finally, let $I_H : \{ \mathbb{F}_q^n \} \to \mathbb{Z}_{q^n}^{q^r}$ be defined as
$$ I_H(\cS) = (|\cS_{H,1}|,|\cS_{H,2}|,\dotsc,|\cS_{H,{q^r-1}}|,|\cS_{H,0}|), $$
for $\cS \subseteq \{ \mathbb{F}_q^n \}$. We provide the following example illustrating these definitions.

\begin{example}\label{ex:mapping} Suppose $q=2, n=3$, $\cS = \{ (0,0,0)$, $(1,1,0)$, $(1,0,1)$, $(0,0,1) \}$. We represent the elements of $\mathbb{F}_4$ as $\alpha^1 = (0,1)^T$, $\alpha^2 = (1, 1)^T$, $\alpha^3 = (1, 0)$, and $\alpha^4=(0, 0)^T$ and let \[H= {\left(\alpha^3\ \alpha^1\ \alpha^2\right)=}\left( \begin{array}{ccc}
1 & 0 & 1 \\
0 & 1 & 1  \end{array} \right).\] 
We have $I_{H}(\cS) = (1, 2, 0, 1)$. In this case, 
\begin{align*}
\cS_{H,1}&=\{(1,0,1)\},\\
\cS_{H,2}&=\{ (0,0,1), (1,1,0) \},\\
\cS_{H,4}&=\{(0,0,0)\}
\end{align*}
and $\cS_{H,2,2}=(1,1,0).$
 \end{example}

To describe the encoding (and subsequent decoding) procedure, we also make use of the following matrices: 
\begin{enumerate}
\item $H_{\ell} \in \mathbb{F}_q^{r \times n}$, for some positive integer $r$, is the parity check matrix for an $[n, 2\ell+1 ]_{q}$ code $\cC_{\ell}$.
\item $H_\comp \in \mathbb{F}_2^{u \times q^r}$, for some positive integer $u$, is the parity check matrix for a $[q^r, 2h+1]_2 $ code $\cC_\comp $.
\item $H_{F} \in \mathbb{F}_q^{n \times n}$, and $\bar{H}_{\ell} \in \mathbb{F}_q^{(n-r) \times n}$ are such that
$H_{F}=\left( \begin{array}{c}
H_{\ell} \\
\bar{H}_{\ell}
\end{array} \right)$
has full rank.
\end{enumerate}

In addition to these, we will require one more tool to encode $\bfw_2$. We first introduce some additional notation similar to \cite{KBE11}. Let $\bfb=(b_1, b_2, \ldots, b_m)$ be a sequence of length $m$ with elements from $\mathbb{F}_{q^{n-r}}$ such that for any $\bfa\in\{0,1\}^m$ with at most $s$ nonzero entries, $\bfa\cdot\bfb\neq0.$
% positive integer $k$ where $k \leq s$, we have \tr{?} $$ \sum_{j=1}^{k} a_j b_{i_j} \neq  0$$ where $1 \leq i_1 < i_2 < \cdots < i_k \leq m$ are distinct and \tr{any?} $\{a_1, a_2, \ldots, a_k \} \subseteq \{ -1, 1 \}$. 
Then, we refer to the sequence $\bfb$ as a $B_{s}$ sequence. Notice that a $B_s$ sequence can be formed from the columns of a parity check matrix for an $[m,d]_q$ code with dimension $n-(n-r)$ provided $d \geq s+1$.

We now proceed by describing the encoding procedure followed by the decoding procedure.

%We assume that $\gf(2^{r})$ is a subfield of $\gf(2^{b-r})$. 

\subsection{Encoding}\label{subsec:enc}

The following procedure is performed on both Host $A$ and Host $B$ but the notation corresponds to Host $A$. We assume that $m \geq q^r$ and that $\bfb=(b_1, b_2, \ldots, b_m)$ is a $B_{h}$ sequence. In Lemma~\ref{lem:encodecorrect}, we give a sufficient condition for {the existence of such $\bfb$ of length} $m \geq q^r$ so that the encoding procedure executes correctly.
\begin{enumerate}
\item Let $\bfz^A = I_{H_{\ell}}(\cS^A)\bmod 2$.
\item Define $\bfw_1^A = H _\comp  \cdot \bfz^A$ {(in $\mathbb{F}_2$)}.
%\item We define $\bfh=(h_\comp , \ldots, h_{2^r}) \in \mathbb{Z}^{2^r}_{2^{b-r}}$ so that for $1 \leq i \leq 2^r$, let $h_i = F^{-1}_{b-r}(\bar{H}_s \cdot \sum \cS_{H,i})$.
\item Let 
\begin{align*}
\bfw_2^A &= \sum_{i=1}^{q^r} b_i \cdot \sum_{j=1}^{|\cS_{H_{\ell},i}|} \bar{H}_\ell \cdot {(\cS^A)}_{H_{\ell},i,j}\\
&=\sum_{\bfx\in\cS^A}b_{\log_\alpha(H_\ell\cdot \bfx) }\cdot\bar H_\ell\cdot \bfx
\end{align*}
where the computations are performed over $\mathbb{F}_{q^{n-r}}$ with $n-r > r$.
%\item For $1 \leq i \leq n$, let $\bfw_2 = \sum_{k=1}^{2^r} \sum \cS_{H_\comp ,i}$
\end{enumerate}
Then Host $A$ transmits $(\bfw_1^A, \bfw_2^A)$ to Host $B$. Similarly, Host $B$ computes and transmits $(\bfw_1^B,\bfw_2^B)$. Then decoding at each host is performed based on $\bfw_1=\bfw_1^A+\bfw_1^B$ and $\bfw_2=\bfw_2^A+\bfw_2^B$. Below, to give intuition to the encoding procedure, we consider these quantities, in addition to $\bfz=\bfz^A+\bfz^B$ (in $\mathbb{F}_2$). Further details will be presented in Subsection~\ref{ssc:dec}, which describes the decoding procedure.

To motivate the our encoding algorithm, first, note that points that lie in the intersection of $\cS^A$ and $\cS^B$ contribute to both $\bfz^A$ and $\bfz^B$, and so their contribution to $\bfz$ cancels out. Hence,
\begin{equation}\label{eq:z}
\bfz=I_{H_\ell}(\cS^A\bigtriangleup\cS^B).
\end{equation}
Thus $\bfz$ contains only information that are relevant to reconciliation. The relationship between $\bfz$ and $\bfw_1$,
\begin{equation}\label{eq:w1}
\bfw_1 =H_\comp\cdot(\bfz^A+\bfz^B)=H_\comp\cdot\bfz,
\end{equation}
is that of compression since $u\ll q^r$ and furthermore we will show later that given $\bfw_1$, the decoder can compute $\bfz$.

For $\bfw_2=\bfw_2^A+\bfw_2^B$, since the characteristic of the field $\mathbb{F}_{q^{n-r}}$ is 2, we have
\begin{equation}\label{eq:w2}
\bfw_2= \sum_{\bfx\in\cS^A\bigtriangleup\cS^B}b_{\log_\alpha(H_\ell\cdot \bfx) }\cdot\bar H_\ell\cdot \bfx.
\end{equation}

We next present an example for the encoding and then we give a lemma providing a sufficient condition for the existence of a $B_h$ sequence required for the encoding. In Subsection~\ref{ssc:dec}, present the details of recovering $\cS^A\bigtriangleup\cS^B$ from $\bfw_1$ and $\bfw_2$.

%\added{To gain intuition, let us consider the special case in which $\cS^A=\emptyset$ and $\cS^B$ contains at most $h$ vectors with pairwise distances at most $\ell$. Since $H_\ell$ has minimum distance $2\ell+1$, $I_{H_{\ell}}(\cS)$ is a binary vector, and thus, $z=I_{H_{\ell}}(\cS)$. Furthermore, weight of}

%The matrices $H_{\ell}$ and $H_\comp $ are used together to form the $\bfw_1$ portion of the information transmitted between Host $A$ and Host $B$. The matrix $\bar{H}_{\ell}$ is used to form the $\bfw_2$ portion. 

%Let $F_{UB} = n +   \frac{\ell}{2}(h-1) \log_2(n)$ so that $F_{UB}$ represents approximately the number of bits of information exchange required by our transmission scheme when the sets $\cS^A, \cS^B$ are binary vectors and let $F_{LB} = n + \frac{\ell}{2}\log_2(n) \left( h-2 \right)$ be the approximate lower bound on the number of bits of information exchange from (\ref{eq:approx1}). It is straightforward to see that when $h-1 \approx \log_2(n)$, we have $ \frac{F_{UB}}{F_{LB}} = 1$ so that under these conditions the number of bits transmitted between Host $A$ and $B$ is asymptotically optimal.

\begin{lemma}\label{lem:encodecorrect} For any integers $n,\ell,h$, there exists a $B_h$ sequence over $\mathbb{F}_{2^{n-r}}$ of length $m \geq q^r$ if 
	$$\mathbb{F}_{n^{(2\ell+1)h}} \subseteq \mathbb{F}_{q^{n-r}} ,$$
	and $h+1 \leq n^{2\ell+1}$.
\end{lemma}
\begin{IEEEproof} We prove the result by considering a code of length $M$ where $M = q^r$. First, notice that by the BCH lower bound  $r \leq 2 \ell \lceil \log_q (n) \rceil$ \cite{R06} and so we assume  $r = 2 \ell \lceil \log_q (n) \rceil$. In this case, we have that $M = n^{2\ell+1}$. We make use of an extended Reed-Solomon code of length $M$ with minimum distance $h+1$. An extended Reed-Solomon code of this length and minimum distance has a parity check matrix $H_{RS} \in \mathbb{F}_M^{h \times M}$ of dimension $h$. Interpreting each column vector of $H_{RS}$ as unique element from $\mathbb{F}_{M^h}$ (using an injective mapping similar to Example~\ref{ex:mapping}), we set the elements in $B_h$ to be the elements from $\mathbb{F}_{M^h}$ that correspond to columns of $H_{RS}$. If $\mathbb{F}_{M^h} \subseteq \mathbb{F}_{q^{n-r}}$ then the elements in $B_h$ are also from the field $\mathbb{F}_{q^{n-r}}$ and so the result follows.
\end{IEEEproof}

We note that for the case where $q=2$, we can strengthen Lemma~\ref{lem:encodecorrect} by using the BCH bound for binary codes. This is given in the next claim.

\begin{claim}\label{cl:encodecorrectbinary} For any integers $n,\ell,h$ and $q=2$, there exists a $B_h$ sequence over $\mathbb{F}_{2^{n-r}}$ of length $m \geq q^r$ if 
	$$\mathbb{F}_{n^{\ell h}} \subseteq \mathbb{F}_{2^{n-r}},$$
	and $h+1 \leq n^{\frac{2\ell+1}{2}}$.
\end{claim}

\subsection{Decoding}\label{ssc:dec}

Suppose $(\bfw^A_1, \bfw^A_2)$ is the information transmitted by Host $A$ to Host $B$ and suppose $(\bfw^B_1, \bfw^B_2)$ is the result of the encoding procedure if it is performed on Host $B$. We illustrate how to recover $\cS^A \bigtriangleup \cS^B$ given $(\bfw^A_1, \bfw^A_2)$, $(\bfw^B_1, \bfw^B_2)$. 

We first describe in words the ideas behind the decoding procedure. The decoding procedure has two broad stages whereby, in the first stage we determine the locations of the elements in $\cS^A \bigtriangleup \cS^B$ relative to some $\bfx \in \cS^A \bigtriangleup \cS^B$ and then in the second stage the element $\bfx$ is recovered. The decoding begins by first recovering the syndromes of the elements in the set $\cS^A \bigtriangleup \cS^B$. More precisely, as a result of the error-correction ability of the code with a parity check matrix $H_\comp $, we first recover the set of syndromes $\cS_S = \{ H_\ell \cdot \bfy : \bfy \in \cS^A \bigtriangleup \cS^B \}$. Next, we arbitrarily choose an element say $\bfx_S \in \cS_S$. Given this setup, $\bfx$ (described earlier) is precisely equal to the element which maps to $\bfx_S$ under the map $H_\ell$, i.e., $\bfx_S = H_\ell \cdot \bfx$.

To determine the locations of the other elements in $\cS^A \bigtriangleup \cS^B$ relative to $\bfx$ we add every element in the set $\cS_S$ to $\bfx_S$. Let $\cS_L = \{ \bfy_S + \bfx_S : \bfy_S \in \cS_S{\setminus\bfx_S} \}$. As will be described below in more detail, from the set $\cS_L$ we can determine the values of the elements in $\cS^A \bigtriangleup \cS^B$ relative to $\bfx$. Next, the value of $\bfx$ is determined by canceling out some of the contributions of the elements in $(\cS^A \bigtriangleup \cS^B) \setminus \bfx$ from the vector $\bfw_2$. 

We now describe in more details the procedure before proving its correctness.
%We assume the procedure below is taking place on Host $B$.
Suppose $\cD_\comp : \mathbb{F}_2^{u} \to \mathbb{F}_2^{q^{r}}$ is the decoder for the code $C_\comp $ which by assumption has minimum Hamming distance at least $2h+1$. $\cD_\comp$ takes as input a syndrome and outputs an error vector with Hamming weight at most $h$. Let $\cD_\ell : \mathbb{F}_q^r \to \mathbb{F}_q^n$ be the decoder for $\cC_\ell$, which has Hamming distance $2\ell+1$.  The decoder $\cD_\ell$ takes as input a syndrome and outputs an error vector with Hamming weight at most $\ell$. In the following, $\alpha$ is a primitive element of $\mathbb{F}_q^r$. {The decoding algorithm is presented next:}

%We refer to the binary vector of length $n$ with all-zeros except in position $j$ as $\bfe_j^n$. 

\begin{enumerate}
\item Let $\hat{\bfz} = \cD_\comp(\bfw^A_1 + \bfw^B_1)$.
\item Suppose $\hat{\bfz}$ has $1$s in positions $\{ k_1, k_2, \ldots, k_v \}$. If $\hat{\bfz}=\bf0$, then let $ \cF = \emptyset$, and stop. 
%\item Let $\bfs' = F_{r}^{-1}(k_1) - F_{r}^{-1}(k_2)$.
\item Define $\hat{\bfe}_2=\cD_\ell(\alpha^{k_1} + \alpha^{k_2} ), \hat{\bfe}_3=\cD_\ell(\alpha^{k_1} + \alpha^{k_3} ), \ldots, \hat{\bfe}_{v}=\cD_\ell(\alpha^{k_1} + \alpha^{k_{v}} )$.
\item Let $\bfz' = \bfw_2^A + \bfw_2^B + \sum_{i=2}^{v} b_{k_i} \cdot \bar{H}_\ell \cdot \hat{\bfe}_i$.
\item Define $\bfs_2 = \bfz' / (b_{k_1} + b_{k_2} + \ldots + b_{k_v})$.
%\item If there exists a $\hat{\bfy}$ is such that $\hat{\bfy} \in \cS^B$, $H_{\ell} \cdot \hat{\bfy} = \alpha^{k_1}$, and $\bar{H}_{\ell} \cdot \hat{\bfy} = \bfs_2$, proceed to step 8). If no such $\hat{\bfy}$ exists, then go to step 7).
%\item Let $\bfz'' = \bfw_2^B + \bfw_2^A + \beta^{k_1} \cdot \bar{H}_\ell \cdot \hat{\bfe}$. 
%\item Let $\bfs_2' = \bfz' / (\beta^{k_1} + \beta^{k_2})$.
\item Let $\hat{\bfx}= H_{F}^{-1} \cdot 
(\alpha^{k_1}, \bfs_2)^T$.
\item $\cF = \{ {\hat{\bfx}}, \hat{\bfx} + \hat{\bfe}_2, \ldots, \hat{\bfx} + \hat{\bfe}_{v}  \}$.
\end{enumerate}

As shown below, the vector $\hat{\bfz}$ essentially gives us the set $\cS_S$ mentioned before. Furthermore, $\bfx_S=\alpha^{k_1}$ and $\{\alpha^{k_1}+\alpha^{k_2},\dotsc,\alpha^{k_1}+\alpha^{k_v}\}$ corresponds to $\cS_L$ described earlier.
\begin{claim}\label{cl:introcl} At the end of step 1) of the decoding $\hat{\bfz} = I_{H_\ell}(\cS^A \bigtriangleup \cS^B )$. \end{claim}
\begin{proof} From \eqref{eq:w1}, we have $\bfw_1=\bfw_1^A+\bfw_1^B=H_\comp \cdot \bfz$. The minimum distance of $C_\comp$ is at least $2h+1$ and the weight of $\bfz$ is at most $h$. Thus
	\[\hat{\bfz}=\cD_\comp(\bfw_1)=\bfz=I_{H_\ell}(\cS^A \bigtriangleup \cS^B ).
	\] 
\end{proof}

\begin{theorem}\label{th:1sets} $\cF = \cS^A \bigtriangleup \cS^B$ when $(\cS^A, \cS^B)$ are $(1,h,\ell)$-sets. \end{theorem}
\begin{proof} Suppose $\cS^A \bigtriangleup \cS^B = \{ \bfx_{1}, \ldots, \bfx_{T} \},$ where $T \leq h$. Since $(\cS^A, \cS^B)$ are $(1, h, \ell)$-sets, we can write 
\begin{align}\label{eq:setdiffl1}
\cS^A \bigtriangleup \cS^B=\{ \bfx, \bfx+\bfe_2, \ldots, \bfx + \bfe_{T}  \}
\end{align}
where $\bfx$ is any element in $\cS^A \bigtriangleup \cS^B$ and for $2 \leq i \leq T$, $\wt(\bfe_i) \leq \ell$. Furthermore, since $|\cS^A \bigtriangleup \cS^B |=T$, for any distinct $2 \leq i,j \leq T$, we have $\bfe_i \neq \bfe_j$. 

Let \(k_1=\log_\alpha (H_\ell\cdot\bfx)\) so that $H_\ell\cdot\bfx=\alpha^{k_1}.$ Suppose that $\bfx = \bfc + \bfe$ for some $\bfc \in \cC_\ell$ such that $\wt(\bfe)$ is minimized (i.e., there does not exist some other $\bfc' \in \cC_\ell$ where $\wt(\bfc' + \bfx) < \wt(\bfe)$). Then, we have
\[ H_\ell \cdot (\cS^A \bigtriangleup \cS^B)  = \{ H_\ell \cdot \bfe, H_\ell \cdot ( \bfe + \bfe_2), \ldots, H_\ell \cdot ( \bfe + \bfe_{T})  \}.\]

Recall from the previous discussion that $\bfe_i$ are distinct, nonzero, and have weight at most $\ell$, for $2\le i\le T$. It follows that 
%for any $j\neq i$, $H_\ell \cdot \bfe_i \neq H_\ell \cdot \bfe_j$. Thus, 
$|H_\ell \cdot (\cS^A \bigtriangleup \cS^B)| = | \cS^A \bigtriangleup \cS^B | = T$. For $2\le i\le T$, let $k_i = \log_\alpha(H_\ell\cdot(\bfe+\bfe_i))$.
From Claim~\ref{cl:introcl}, we have that at step 2) $\hat{\bfz}$ is non-zero in positions $\alpha^{k_1}, \alpha^{k_2}, \ldots, \alpha^{k_{v}}$, where $v=T$.
% where
%\begin{align}\label{eq:setdiffl3}
%H_\ell \cdot (\cS^A \bigtriangleup \cS^B) =  \{ \alpha^{k_1}, \alpha^{k_2}, \ldots, \alpha^{k_{v}} \}
%\end{align}
%and $v=k$. Suppose that $\alpha^{k_1} = H_\ell \cdot \bfx, \alpha^{k_2} = H_\ell \cdot (\bfx + \bfe_2), \alpha^{k_3} = H_\ell \cdot (\bfx + \bfe_3)$, and so on.

Notice that if $\hat{\bfz}=\bf0$, then $\cS^A \bigtriangleup \cS^B = \emptyset$. Since $H_\ell \cdot \bfx = \alpha^{k_1}$ and $\cD_\ell$ can correct up to $\ell$ errors, at step 3) of the decoding $\hat{\bfe}_2 = \cD_\ell(\alpha^{k_1} + \alpha^{k_2}) = \cD_\ell(H_\ell \cdot \bfx + H_\ell \cdot ( \bfx + \bfe_2)) = \cD_\ell(H_\ell \cdot \bfe_2) = \bfe_2$ and similarly $\hat{\bfe}_3 = \bfe_3, \ldots, \hat{\bfe}_v = \bfe_v$. 

From \eqref{eq:w2}, 
\begin{align*}
\bfw_2&=\bfw_2^A + \bfw_2^B =b_{k_1} \cdot \bar{H}_\ell \cdot \bfx + \sum_{i=2}^v b_{k_i} \cdot \bar{H}_\ell \cdot( \bfx + \bfe_i)\\ 
&= \left( \sum_{i=1}^v b_{k_i} \right) \cdot \bar{H}_\ell \cdot \bfx + \sum_{i=2}^v b_{k_i} \cdot \bar{H}_\ell \cdot \bfe_i.
%\sum_{i=1}^{2^r} \alpha^i \cdot \sum_{j=1}^{|\cS^B_{H,i}|} \bar{H}_\ell \cdot \cS^B_{H,i,j} + \sum_{i=1}^{2^r} \alpha^i \cdot \sum_{j=1}^{|\cS^A_{H,i}|} \bar{H}_\ell \cdot \cS^A_{H,i,j}\label{eq:lem11} \\
%&=\sum_{i=1}^{2^r} \alpha^i \cdot \psi\left(\bar{H}_\ell \cdot \sum \cS^B_{H,i} + \bar{H}_\ell \cdot \sum \cS^A_{H,i} \right)\\
%&=b_{k_1} \cdot \bar{H}_\ell \cdot \bfx  + \beta^{k_2} \cdot \bar{H}_\ell \cdot (\bfx + \bfe) \\
%&=(\beta^{k_1} + \beta^{k_2} ) \cdot \bar{H}_\ell \cdot \bfx + \beta^{k_2} \cdot \bar{H}_\ell \cdot \bfe.
\end{align*}
Then, at step 4) of the decoding we have 
\begin{align}\label{eq:zbart=1}
\bfz' &= \bfw_2^A + \bfw_2^B + \sum_{i=2}^{v} b_{k_i} \cdot \bar{H}_\ell \cdot \hat{\bfe}_i =  \left( \sum_{i=1}^v b_{k_i} \right) \cdot \bar{H}_\ell \cdot \bfx.
\end{align}
Since $T=v \leq h$ and $\bfb=(b_1, b_2, \ldots, b_m)$ is a $B_{h}$ sequence, we have $\bfs_2=\bar{H}_\ell \cdot \bfx$ at step 5) of the decoding. Then, at step 6), $\hat{\bfx} = \bfx$ since $\alpha^{k_1} = H_\ell \cdot \bfx$ and $H_{F}$ has full rank. At step 7), since $\hat{\bfe}_2 = \bfe_2, \hat{\bfe}_3 = \bfe_3, \ldots, \hat{\bfe}_T = \bfe_T$, we have $\cF = \cS^A \bigtriangleup \cS^B$ as desired.
\end{proof}

We note that $(\bfw_1, \bfw_2)$ requires approximately $u + \lg q^{n-r} = u + (n-r)\lg q$ bits where $u$ is the dimension of the parity check matrix for $H_\comp $. If we approximate $u = hr \lg q$ and $r= \ell (\lg n + \lg q)$, then $(\bfw_1, \bfw_2)$ requires approximately $\lg q \left( n +  (h-1) \ell (\lg n + \lg q ) \right)$ bits of information. If $q=2$, $(\bfw_1, \bfw_2)$ requires approximately 
\begin{align}\label{eq:lower1}
n + (h-1) \ell (\lg n + 1).
\end{align}
If the approach from \cite{MT02} were used then at least $h (n + 1)$ bits of information exchange would be required which is significantly more than the quantity in (\ref{eq:lower1}). {Nevertheless, the upper bound on information exchange given in Theorem~\ref{th:asmptbnds} for $t=1$ is $\le 2h\ell(H_q(\lambda)/\lambda)$ which for a constant $\lambda$ is asymtotically smaller than the information exchange in (\ref{eq:lower1}), and so achieving optimality is still an open problem.}

Note that the basic approach taken in this section was to first determine the differences between elements in the symmetric difference. Then, the idea was to use those differences so that at step 4) of the decoding, we produced a $\bfz'$ which is basically a scaled version of $\bfs_2$. Note that we can obtain $\bfz'$, as shown in (\ref{eq:zbart=1}), by multiplying $\bfb$ (which is a $B_h$ sequence) by a vector, say $\bfu$ with at most $h$ identical non-zero components so that $\bfu$ has rank at most $1$. In the next section, we extend this idea by showing how, given the relationship between elements in the symmetric difference which are close to each other, we can recover the symmetric difference for certain classes of $(t,h,\ell)$-sets by solving for a low-rank vector.

\section{Reconciling Certain Classes of $(t,h,\ell)$-sets}\label{sec:t>1}
\newcommand{\minusI}{\bar{I}}
\newcommand{\Cen}{\mathsf{Cen}}
\newcommand{\rk}{\operatorname{rk}}
\newcommand{\ball}{\cB}

In this section, we will detail an approach that addresses the case where $t \geq 1$. First we give an overview of our method with a fair amount of detail, but postpone the proofs until later in the section. In \Cref{ssc:t>1Enc} and  \Cref{ssc:t>1Dec}, we formally present the encoding and decoding algorithms.

Let us now fix some notation. For a set $I \subseteq [n]$, and a vector $\bfx \in \mathbb{F}_2^n$, let $\bfx_I$ denote the vector that results by discarding the components of $\bfx$ outside $I$. For example, if $\bfx = (1,0,1,0)$, then $\bfx_{\{1,3\}} = (1,1)$. For a set of vectors $S \subseteq \mathbb{F}_2^n$, let $S_I$ denote the set of vectors that results from discarding the components of each element in $S$ outside $I$. 

The aim of this section is to describe an approach to synchronize two sets of data $\cS^A$ and $\cS^B$ where the symmetric difference between the sets has the following structure. For given $t$, $h$, and $\ell$,
$
 \cS^A \bigtriangleup \cS^B =\bigcup_{i=1}^{j} {\ball_i}, 
$
where $\ball_i=\{ \bfx_{i,1}, \ldots, \bfx_{i,k_i} \}$ and the following hold:
\begin{enumerate}
\item $j \leq t$; 
\item for $1 \leq i \leq j$, $k_i \leq h$;
%\item $|\cS^A \bigtriangleup \cS^B| = 2m \leq t$, 
\item for any $\bfu, \bfw \in \ball_i$, we have $d_H(\bfu, \bfw) \leq \ell$; and
\item $\exists I \subseteq [n]$ such that:
\begin{enumerate}
\item  For any $ i_1\neq i_2$ and any $\bfx\in\ball_{i_1}, \bfy\in\ball_{i_2}$, we have $\bfx_I \neq \bfy_I$; and
\item For any $\bfx, \bfy \in \ball_i$, we have $\bfx_I = \bfy_I$.
\end{enumerate}
\end{enumerate}
Each set $\cB_i$ is referred to as a {\textit{\textbf{(difference) block}}}.

If the sets $\cS^A$ and $\cS^B$ satisfy the conditions $1) -4)$, then (with a slight abuse of notation) $\cS^A$, $\cS^B$ are called $(t,h,\ell)$-sets. For the remainder of the paper, we assume $(t,h,\ell)$-sets are as defined in this section and not in Definition~\ref{def:thlset}. 

We note that our original definition for $(t,h,\ell)$-sets in~Definition~\ref{def:thlset} was for general $q$-ary strings, and it did not require condition 4). We chose to focus on binary strings for simplicity of presentation, but the ideas extend to the more general case. The reason for including condition 4) here is that it is used during the encoding/decoding to group the elements in $S^A \bigtriangleup S^B$. For the case where $t=1$, which was considered in the previous section, there exists at most one difference block and so no grouping was required. 

%It remains an open problem on how to design a transmission scheme without condition 4) for $t > 1$.

Assuming that the elements of the symmetric difference are chosen at random but with the constraint that conditions 1)-3) are satisfied, a simple argument shows that condition 4) is violated with probability at most
\[\frac{th^2\ell|I|}{n}+\frac{t^2h^2}{2^{|I|}}\]
which approaches 0 for $|I|=\lg(n)$ provided that $t^2h^2\ell=o(n/\lg n)$. Furthermore, given the setup where database documents are being synchronized, the set $I$ could be derived from the document's unique identifier for instance. Thus, the $(t,h,\ell)$-sets considered here could arise in several different ways. 

Before continuing, we revisit the example from Section~\ref{sec:model} of a $(t,h,\ell)$-set {in the context of the new definition}.
  
\begin{example} {Suppose $\cS^A, \cS^B \in \mathbb{F}_2^5$, where 
\begin{equation*}
\begin{split}
\cS^A&=\{(0,0,0,0,0), (1,0,1,1,1)\},\\
\cS^B&=\{(0,0,0,0,0), (1,1,0,0,1)\}. 
\end{split}
\end{equation*}
Then we say that $(\cS^A, \cS^B)$ are $(1,2,3)$-sets since
\[\cS^A \bigtriangleup \cS^B=\{ (1,\textcolor{red}{0},\textcolor{red}{1},\textcolor{red}{1},1),(1,\textcolor{red}{1},\textcolor{red}{0},\textcolor{red}{0},1) \},\] which can be decomposed into $1$ set of size $2$ whereby the Hamming distance between any two elements is at most $3$.} Notice here that $I=\{1,5\}$.  \end{example}

Next, we give an overview of our method with a fair amount of detail, and postpone the formal presentation of the encoding and decoding algorithms to \Cref{ssc:t>1Enc,ssc:t>1Dec}, respectively. Similar to the algorithm for $(1,h,\ell)$-sets from Section~\ref{sec:r1sets}, and as discussed previously, the process of synchronizing $(t,h,\ell)$-sets will be broken down into $2$ main stages:
\begin{itemize}
\item[]\hspace{-1.4em} Stage 1) Determine the differences between the elements in the symmetric difference.
\item[]\hspace{-1.4em} Stage 2) Recover the elements in the symmetric difference.
\end{itemize}

%In order to accomplish 1), we will require an error-correcting code to identify regions of the space $\gf(2)^n$ where the covering radius of the code is $\lfloor \frac{\ell}{2} \rfloor$. Let $\cC_{lab} \subseteq \mathbb{F}_2^n$ be such a code where the subscript $lab$ refers to the fact that we will use this code to ``label'' the space $\mathbb{F}_2^n$. For a vector $\bfv \in \gf(2)^n$, let $\cB_\ell(\bfv)$ refer to the set of vectors possible given that at most $\ell$ components of $\bfv$ may have their values changed. Since the covering radius of the code $\cC_{lab}$ is $\lfloor \frac{\ell}{2} \rfloor$, we can write any vector $\bfz \in \gf(2)^n$ uniquely as $\bfz = \bfx + \bfe$ where $wt(\bfe) \leq \lfloor \frac{\ell}{2} \rfloor$. Furthermore, if two vectors $\bfy_1, \bfy_2 \in \gf(2)^n$ are such that $\bfy_1 \in \cB_{ \lfloor \frac{\ell}{2} \rfloor }(\bfx)$ and $\bfy_2 \in \cB_{ \lfloor \frac{\ell}{2} \rfloor }(\bfx)$, then $d_H(\bfy_1, \bfy_2) \leq \ell$.

%We first consider stage 1). The main idea behind the first stage of the encoding/decoding is to use two functions (to be defined later) $f:\gf(2)^{|I|}\to\gf(Q)$ and $M:\gf(2^n)\to[N]$ to represent elements (at most $t \cdot h$) in the symmetric difference and, given this representation to determine the differences between the elements in the symmetric difference. 

Notice that under our setup, there exists a set $\Cen \subseteq \mathbb{F}_2^n$ of size at most $t$ containing one element $\bfc_i$ from each difference block $\ball_i$ such that for each $i$ and any $\bfy \in \cB_i$, we have $d_H(\bfy,\bfc_i)\le\ell$.
%$\bfy \in \bigtriangleup \cS^B$, $\exists \bfx \in \Cen$ we have $$ d_H(\bfx, \bfy) \leq \ell. $$
We refer to the set $\Cen$ as the \textit{\textbf{center set}} and to each $\bfc_i$ as a \textit{\textbf{block center}}.

%Furthermore, we say that the two elements $\bfv_1, \bfv_2$ are in the same \textit{\textbf{difference block}} \tr{($\ball_i$?)} \textcolor{blue}{Yeah, it's a good idea to label the sets for later.} if $ d_H(\bfv_1, \bfv_2) \leq \ell$ and $\exists \bfu \in \Cen$ where both $d_H(\bfv_1, \bfu) \leq \ell$ and $d_H(\bfv_2, \bfu) \leq \ell$. 

Our goal during Stage 1) will be to recover the differences between the elements in each block. To this end, we represent the information in the sets as length-$N$ vectors over $\mathbb{F}_Q$, denoted $\bfz_{1} = (z_{1,1}, \ldots, z_{1,N}) \in \mathbb{F}_Q^N$,  where $\mathbb{F}_Q$ has characteristic two. The values of $Q$ and $N$ are chosen to ensure the existence of three maps 
\begin{align*}
M &: \mathbb{F}_2^n \to [N],\\
E &: [N] \times [N] \to \mathbb{F}_2^n, \\
f &: \mathbb{F}_2^{|I|} \to \mathbb{F}_Q
\end{align*}
with certain properties that will be described shortly. Both maps are also used in the second stage of synchronization.

The map $M$ is a function that will be used to assign to each $\bfx\in\cS$ (where $\cS=\cS^A$ or $\cS=\cS^B$) a position in $\bfz_1$, that is, $M(\bfx)\in[N]$. This assignment satisfies the following property.
\begin{property}\label{prpr:M}
	The map $M$ is such that if $\bfx_1, \bfx_2 \in \cB_i$ for some $i \in [t]$, then $\bfx_1$ and $\bfx_2$ are mapped to different positions, i.e., 
	\begin{equation}\label{eq:dist}
	M(\bfx_1)\neq M(\bfx_2).
	\end{equation}
\end{property} 
As a result, no two elements belonging to the same difference block are mapped to the same position. 

The map $E$, which will be useful for determining the differences between elements in the symmetric difference, has the following property.
\begin{property}\label{prpr:E}
	The map $E$ is such that if $\bfx_1, \bfx_2 \in \cB_i$ for some $i \in [t]$, then 
	\begin{equation}\label{eq:emap}
	E(M(\bfx_1), M(\bfx_2)) = \bfx_1 + \bfx_2.
	\end{equation}
\end{property}

We now turn to discussing the map $f$. For now, we assume this map has the following property. In \Cref{ssc:t>1Enc}, we show how to construct such maps.
\begin{property}\label{prpr:f}
The map $f$ is an \emph{invertible} function such that for any $ \cX \subseteq \mathbb{F}_2^{|I|}, |\cX| \leq 2t$, we have
\begin{align}\label{eq:propFMap}\vspace{-.5ex}
\sum_{\bfx \in \cX} f (\bfx) \neq 0.\vspace{-.5ex}
\end{align}
\end{property}

For a subset $\mathcal{S} \subseteq \mathbb{F}_2^n$ (in particular $\cS=\cS^A$ or $\mathcal{S}=\cS^B$), let
$$\cS_{M,i} = \{ \bfx \in \mathcal{S} : M( \bfx) = i \}.$$ 
The vector $\bfz_1=(z_{1,j})_{j\in[N]}$ is defined as 
\begin{equation}\label{eq:z1}\vspace{-.5ex}
z_{1,j} = \sum_{\bfx \in \cS_{M,j} }   f(\bfx_I).\vspace{-.5ex}
\end{equation}
%
% 
%In a manner analogous to the approach for $(1,h,\ell)$-sets, the parameter $N$ is chosen so that given a pair of distinct elements $\bfx, \bfy \in \cS^A \bigtriangleup \cS^B$, we can determine $\bfx + \bfy$ based on the locations of where $\bfx, \bfy$  are hashed to in our vectors.
%For the case of $t=1$, there is  one difference block and the center set has size one.  As a result, it can be shown that any position in the length $N$ vector can have at most one element mapped to it under the function $M$. However, for the current setup where $t\geq 1$, there may be multiple elements (at most $t$) mapped to the same position in $\bfz_1$. 
%The alphabet size $Q$ is chosen so that for any position in $\bfz_1$, 
The result of Property~\ref{prpr:f} is that for $d \leq t$ and $\bfx^{(1)}_I, \bfx^{(2)}_I, \ldots, \bfx^{(d)}_I\in \cS_{M,j}$,  we can  recover $\bfx^{(1)}_I,\bfx^{(2)}_I, \ldots, \bfx^{(d)}_I$ from their sum $z_{1,j}=\sum_{\bfx \in \cS_{M,j} }   f(\bfx_I)$. 

{Let $\bfz_1^A$ and $\bfz_1^B$ be the result of computing $\bfz_1$ according to~\eqref{eq:z1} on Host $A$ and Host $B$, respectively. Furthermore, let $\dot{\bfz}=\bfz_1^A+\bfz_1^B$. Each host transmits a compressed version of its $\bfz_1$ vector to the other one and so each can then compute $\dot{\bfz}$. The effect of elements in $\cS^A\cap\cS^B$ are canceled out in $\dot{\bfz}$ since they contribute to both $\bfz_1^A$ and $\bfz_1^B$ and $\mathbb{F}_Q$ is an extension field of $\mathbb{F}_2$. Hence,
\begin{equation}\label{eq:zdot}\vspace{-.5ex}
\dot{z}_{j}= \sum_{\bfx \in (\cS^A\bigtriangleup\cS^B)_{M,j} }   f(\bfx_I).\vspace{-.5ex}
\end{equation}}

{Given $\dot\bfz$, from the discussion following~\eqref{eq:z1} and the invertibility of $f$, for each $i$ we can recover the set $\Big \{M(\bfx):\bfx\in\ball_i \Big \}$. Using this information, we can then identify the differences between the elements of each $\ball_i$ using the map $E$, which is the goal of Stage 1). 
{So based on the preceding discussion, from $\dot{\bfz}$, we can find $\bfx_I$ and $f(\bfx_I)$ for each $\bfx\in\cS^A\bigtriangleup\cS^B$, the number of difference blocks $\ball$, the number of elements in each block, and the differences between any two elements in each block.}
%For the first stage of the encoding, we encode the differences between elements and the center set using the vector $\bfz_1$. Suppose, for instance, the encoding is taking place on Host $A$. First, we would initialize $\bfz_1 = (0,\ldots, 0)$. Then, for every $\bfx \in \cS^A$, we would update $\bfz_1$ so that $z_{1,M(\bfx)} = z_{1,M(\bfx)} + f (\bfx_I)$.
 
As mentioned earlier, the hosts do not transmit $\bfz_1^A$ and $\bfz_1^B$ but rather a compressed version of these vectors. Let 
\begin{align}\label{eq:HCr}
H_{\comp} \in \mathbb{F}_Q^{r \times N}
\end{align}
be a parity check matrix for a $[N, 2th + 1]_Q$ code, denoted $\cC_\comp$. Each host computes $\bfw_1 = H_\comp\cdot \bfz_1$ (resulting in $\bfw_1^A$ and $\bfw_1^B$) and transmits it to the other host. So each host can compute $\dot \bfw=\bfw_1^A+\bfw_1^B$. Note that $\dot\bfw=H_\comp \cdot\dot{\bfz}$. Since $\wt(\dot{\bfz})\le th$, the hosts can find $\dot{\bfz}$ using a decoder $\cD_\comp$ for the code $H_\comp$.

For Stage 2), the idea will be to use the differences between a center set and the remaining elements to encode (and subsequently decode) the elements in the center set only. During the decoding, we will produce the symmetric difference given knowledge of a center set and the differences. 

In this stage, we represent our information using the vectors 
\begin{align*}
\bfz^{(0)}_2 &= (z^{(0)}_{2,1}, \ldots, z^{(0)}_{2,N}),\\
%\bfz^{(1)}_2 &= (z^{(1)}_{2,1}, \ldots, z^{(1)}_{2,N}),\\
&\vdots \\
\bfz^{(t-1)}_2 &= (z^{(t-1)}_{2,1}, \ldots, z^{(t-1)}_{2,N})  \in (\mathbb{F}_2^{n-|I|})^N.
\end{align*}  

For shorthand, let $\bar{n} = n - |I|$ and $\minusI = [n] \setminus I$, so that $\bfx_{\minusI}=\bfx_{\left([n]\setminus I\right)}$. Suppose, as before, we are encoding the set $\mathcal{S}$, where $\mathcal{S}=\cS^A$ or $\cS^B$. We let 
\[z^{(k)}_{2,j} = \sum_{\bfx \in \cS_{M,j}} \left(f(\bfx_I)\right)^{2^k} \cdot \bfx_{\minusI}\]
for $k \in \{0,1,\ldots, t-1\}$.
%First, we initialize $\bfz^{(k)}_2 = (0, \ldots, 0)$ for $k \in \{0,1,\ldots, t-1\}$. Then for every $\bfx \in \cS^A$, for $k \in \{0,1,\ldots,t-1\}$, we update $\bfz^{(k)}_2$ so that $z^{(k)}_{2,M(\bfx)} = z^{(k)}_{2,M(\bfx)} + f (\bfx_I)^{2^k} \cdot \bfx_{\minusI}$. 
For this stage we implicitly make use of a bijection between $\mathbb{F}_2^{\bar{n}}$ and $\mathbb{F}_{2^{\bar{n}}}$. {Further, we assume $\mathbb{F}_Q \subseteq \mathbb{F}_{2^{\bar{n}}}$.}

We need another matrix to fully describe the encoding process. Let 
\begin{align}\label{eq:HFmatrix}
H_F \in \mathbb{F}_R^{t \times N},
\end{align}
where $R \geq N^{(2^t-1)h}$, such that the following property holds. 

\begin{property}\label{prpr:solvesys} {For any submatrix $H_F'$ of $H_F$, consisting of any $c\le th$ nonzero columns from $H_F$, and for any $\bfs \in \mathbb{F}_{2^{\bar{n}}}^t$, there exists at most one choice of a vector $\bfv\in \mathbb{F}_{2^{\bar n}}^c$ over $\mathbb{F}_{2^{\bar{n}}}$ with  $\rk(\bfv) \leq t$ that satisfies}\vspace{-.5ex}
\begin{equation}\label{eq:solvesys}
H_F' \cdot \bfv = \bfs.\vspace{-.5ex}
\end{equation}
Here $\rk(\bfv)$ denotes the rank of $\bfv$ over $\mathbb{F}_2$ if $\bfv$ is interpreted as an $\bar n\times c$ matrix.
\end{property}

Given the $t \times N$ matrix $H_F$, Host $A$ constructs
\[\bfw_2^A = (\bfw_2^{A,(0)}, \ldots, \bfw_2^{A,(t-1)})\]
where $\bfw_2^{A,(k)}= H_F \cdot \bfz^{A,(k)}_2$,
and transmits it to Host $B$.

We now turn to describe decoding in Stage~2). For clarity of presentation, we assume that the vector $\dot{\bfz}=\bfz_1^A+\bfz_1^B$ from Stage~1) has the following nonzero elements:
\begin{equation*}
\begin{split}
\dot\bfz_{j_1} &= \bfsigma_1,\\
\dot\bfz_{j_2} &= \bfsigma_1+\bfsigma_2,\\
\dot\bfz_{j_3} &= \bfsigma_2,\\
\end{split}
\end{equation*}
where $\bfsigma_1,\bfsigma_2\in \mathbb{F}_2^{|I|}$. Notice that under this setup, $t=2$ and $h=2$, so that the symmetric difference consists of two blocks, $\ball_1$ and $\ball_2$, each with two elements. Without loss of generality, suppose that $\ball_1=\{X,X+\bfe_1\}$ and $\ball_2=\{Y,Y+{\bfe_2} \}$, where
\begin{align*}
f(X_I)&=\bfsigma_1, &M(X)&=j_1,\\
f((X+\bfe_1)_I)&=\bfsigma_1, &M(X+\bfe_1)&=j_2,\\
f(Y_I)&=\bfsigma_2, &M(Y)&=j_2,\\
f((Y+\bfe_2)_I)&=\bfsigma_2, &M(Y+\bfe_2)&=j_3,
\end{align*}
and where $\wt(\bfe_1)\le\ell$ and $\wt(\bfe_2)\le\ell$. Also note that $(X+\bfe_1)_I=X_I$ and $(Y+\bfe_2)_I=Y_I$. At this point, we still do not know the values of $X$ and $Y$, but from Stage~1) of the decoding we know the values of $\bfe_1$ and $\bfe_2$:
\begin{equation}
\bfe_1=E(j_1,j_2),\quad\bfe_2=E({j_2},j_3).
\end{equation}

Let $\ddot{\bfz}^{(k)}=\bfz_2^{A,(k)}+\bfz_2^{B,(k)}$ and $\ddot{\bfw}^{(k)}=\bfw_2^{A,(k)}+\bfw_2^{B,(k)}$. When decoding, each node can compute $\ddot{\bfw}^{(k)}$, which equals \(H_F \cdot \ddot{\bfz}^{(k)}\).
%\begin{align*}
%\ddot{\bfw}^{(k)}&=H_F \cdot \ddot{\bfz}^{(k)}\\
%& = H_F \cdot (\bfz^{A,(k)}_{2} + \bfz^{B,(k)}_{2})\\
%& = H_F \cdot (\ddot{z}^{(k)}_{1}, \ldots, \ddot{z}^{(k)}_{N})
%\end{align*}

Because we have mapped the elements in $\cS^A \bigtriangleup \cS^B$ to the same locations in both {$\dot{\bfz}$} and $\ddot{\bfz}^{(k)}$ using the map $M$, we know {$\ddot{\bfz}^{(k)}$} has the following nonzero elements
\begin{equation*}
\begin{split}
{\ddot\bfz_{j_1}^{(k)}} &= \bfsigma_1^{2^k}\cdot X_{\minusI},\\
{\ddot\bfz_{j_2}^{(k)}} &= \bfsigma_1^{2^k}\cdot(X+\bfe_1)_{\minusI}+\bfsigma_2^{2^k}\cdot Y_{\minusI},\\
{\ddot\bfz_{j_3}^{(k)}} &= \bfsigma_2^{2^k}\cdot(Y+\bfe_2)_{\minusI},\\
\end{split}
\end{equation*}
for any $k$. At this point, we still do not know the values of $X$ and $Y$ but do know the values of $\bfe_1$ and $\bfe_2$.

The rank $\rk(\ddot{\bfz}^{(k)})$ of $\ddot{\bfz}^{(k)}$ is at most $3$. Our goal now is to decrease the rank of this vector to at most $t=2$ so that we can use Property~\ref{prpr:solvesys}. To do so, from each block $\ball_i$, we arbitrarily pick an element as the block center and as described below, we change every other appearance of an element of $\ball_i$ in $\ddot{\bfz}^{(k)}$ to look like the block center. In our current illustration, we pick the element of $\ball_1$ mapped to $j_1$ and the element of $\ball_2$ mapped to $j_2$ as their respective centers, which we have named $X$ and $Y$.

Let $\bfu \in \mathbb{F}_{2^{\bar{n}}}^N$ be the all-zero vector except in position $j_2$, where $u_{j_2} = \bfsigma_1^{2^k} \cdot (\bfe_1)_{\minusI}$. Notice here that we again implicitly use a bijective mapping between $\mathbb{F}_{2^{\bar{n}}}$ and $\mathbb{F}_2^{\bar{n}}$. We initialize $\hat{S}^{(k)}=\ddot{\bfw}^{(k)}=\bfw_2^{A,(k)}+ {\bfw_2^{B,(k)}}$ and update it by adding $H_F \cdot \bfu$ to it: 
\begin{align*}
\hat{S}^{(k)} &\leftarrow {\ddot{\bfw}^{(k)}} + H_F \cdot \bfu \\
&= H_F \cdot ( \ddot{\bfz}^{(k)} + \bfu). \vspace{-.5ex}
\end{align*}
Note that the $j_2$th position of $\ddot{\bfz}^{(k)} + \bfu$, denoted $(\ddot{\bfz}^{(k)} + \bfu)_{j_2}$, is 
\begin{align*}\vspace{-.5ex}
(\ddot{\bfz}^{(k)} + \bfu)_{j_2} &= \ddot{z}^{(k)}_{j_2} + u_{j_2} \\
&= \Big ( \bfsigma_1^{2^k} \cdot (X + \bfe_1)_{\minusI} + \bfsigma_2^{2^k}\cdot Y_{\minusI}  \Big) + 
%\\&\phantom{{}={}}
\Big ( \bfsigma_1^{\mathclap{2^k}} \cdot (\bfe_1)_{\minusI} \Big )\\
&= \bfsigma_1^{2^k} \cdot X_{\minusI} + \bfsigma_2^{2^k}\cdot Y_{\minusI}.\vspace{-.5ex}
\end{align*}
So now both elements of $\ball_1$ contribute a term of the form $ \bfsigma_1^{2^k} \cdot X_{\minusI}$.

We update $\hat{S}^{(k)}$ again by letting $\hat{S}^{(k)}\leftarrow \hat{S}^{(k)} + H_F\cdot\bfu'$, where $u'_{j_3}=\bfsigma_2^{2^k}\cdot(E(j_2,j_3))_{\bar I}$, so that \vspace{-.5ex}
\begin{align*}
(\ddot{\bfz}^{(k)} + \bfu')_{j_3} &= \bfsigma_2^{2^k} \cdot Y
_{\minusI}. \vspace{-.5ex}
\end{align*}
and that
$$\vspace{-.5ex} \hat{S}^{(k)}= H_F \cdot V^{(k)}, \vspace{-.5ex}$$
where the non-zero entries in $V^{(k)}$ are contained within the set
\begin{align*}\vspace{-.5ex}
U^{(k)}= \{ \bfsigma_1^{2^k} \cdot X_{\minusI}+\bfsigma_2^{2^k} \cdot Y_{\minusI},\  \bfsigma_1^{2^k} \cdot X_{\minusI},\ \bfsigma_2^{2^k} \cdot Y_{\minusI}  \}. \vspace{-.5ex}
\end{align*}

Notice that $\rk(V^{(k)})\le2$ whereas $\rk(\ddot{\bfz})\le3$. Thus it is possible now to use Property~\ref{prpr:solvesys} {to recover $V^{(k)}$ for $k \in \{0,1,2,\ldots,t-1\}$ given that $t=2$ and $h = 2$.  In particular, given:}\vspace{-.5ex}
{\begin{align*}
\bfsigma_1^{2^k}  \cdot X_{\minusI}+\bfsigma_2^{2^k} \cdot Y_{\minusI} \vspace{-.5ex}
\end{align*}}{for $k \in \{0,1\}$ along with knowledge of $\bfsigma_1, \bfsigma_2$ (which we recovered from the first stage of the decoding using the vector $\dot{\bfz}$), we can recover $X_{\minusI}$ and $Y_{\minusI}$, which allows us to determine the center set $\{X, Y\}$. Then, with $X,Y$ and $\dot{\bfz}=\bfz_1^A+\bfz_1^B$ can recover $\ball_1=\{X,X+\bfe_1\}$ and $\ball_2=\{Y,Y+\bfe_2 \}$, so we are able to reconstruct the set $\cS^A \bigtriangleup \cS^B$. }

\subsection{Encoding}\label{ssc:t>1Enc}

In this section, we formally state the encoding algorithm. We present the encoding procedure for Host $A$, but the same applies to Host $B$ as well.

Recall that for a set $S$, we let $S_{M,i} = \{ \bfx \in S : M( \bfx) = i \}$ and that $H_{\comp}$ from (\ref{eq:HCr}) is a parity check matrix for a $[N, 2th+1]_Q$ code, denoted $\cC_{comp}$. Furthermore $H_F$ is as described in (\ref{eq:HFmatrix}). The encoding is as follows.

\begin{enumerate}
	\item Let $\bfz_1^A = (z_{1,1}^A, z_{1,2}^A, \ldots, z_{1,N}^A) \in \mathbb{F}_Q^N$ with
	\[z_{1,j}^A = \sum_{\bfx \in \cS_{M,j}^A }   f(\bfx_I), \]
	and set
	\[\bfw_1^A = H_{\comp} \cdot \bfz_1^A.\]
	\item Let $\bfz^{A,(k)}_2 =(z^{A,(k)}_{2,1}, z^{A,(k)}_{2,2}, \ldots, z^{A,(k)}_{2,N}) \in \mathbb{F}_{2^{\bar{n}}}^N$ with \[z^{A,(k)}_{2,j} = \sum_{\bfx \in \cS_{M,j}^A} \left(f(\bfx_I)\right)^{2^k} \cdot \bfx_{\minusI}\]\vspace{-.5ex}
	for $k \in \{0,1,\ldots, t-1\}$
	and set 
	\vspace{-.5ex}\[\bfw_2^A = ( H_F \cdot \bfz^{A,(0)}_2, H_F \cdot \bfz^{A,(1)}_2, \ldots, H_F \cdot \bfz^{A,(t-1)}_2).\]\vspace{-.5ex}
\end{enumerate}\vspace{-.5ex}
The information $\bfw^A=(\bfw_1^A, \bfw_2^A)$ is then transmitted to Host $B$. The size of $(\bfw^A_1, \bfw^A_2)$ is given in Claim~\ref{cl:size}. 

We now verify the existence of the maps introduced in earlier in this \namecref{sec:t>1}. {The proofs of the next three claims are presented in the appendices.}

%\tg{[!! Is this paragraph now superseded by Property~\ref{prpr:M}? If so, we should remove it.]} Recall that any two elements that map to the same location by $M$ (in its image) belong to different difference blocks. For notational convenience, let $E_\ell$ denote the set of binary strings of weight at most $\ell$. We require that the map $M$ is \sout{linear} and that it has the following property. For any $\bfe_1, \bfe_2 \in E_\ell$ where $\bfe_1 \neq \bfe_2$, 
%\begin{align}\label{eq:MFunc}
%M(\bfe_1) \neq M(\bfe_2). 
%\end{align}

\begin{claim}\label{cl:mmap} There exists a map $M: \mathbb{F}_2^n \to [N]$ satisfying Property~\ref{prpr:M} with $N = n^\ell$. \end{claim}

\begin{claim}\label{cl:emap} There exists a map $E : [N] \times [N] \to \mathbb{F}_2^n$ satisfying Property~\ref{prpr:E}.
\end{claim}

The next claim follows using similar logic by using a parity check matrix with minimum distance $2t+1$ along with a hash function similar to \cite{FKS84}.

%Since $M$ is linear \tr{(?)}, then from (\ref{eq:MFunc}), the following claim holds. 
%
%\begin{claim}\label{cl:dist} For any $i \in [N]$, and any $\bfx, \bfy \in (\cS^A \bigtriangleup \cS^B)_{M,i}$ where $\bfx \neq \bfy$, $d_H(\bfx, \bfy) \geq 2\ell + 1$. \end{claim}

%Recall that in addition to the map $M$, we will also use the map $f : \gf(2)^{|I|} \to \gf(Q)$ in the first part of the encoding.  The map $f$ has the following properties:
%\begin{align}\label{eq:propFMap}
%\sum_{\bfx \in \cX: \cX \subseteq \gf(2)^{|I|}, |\cX| \leq 2t} f (\bfx) \neq 0,
%\end{align}
%\sout{and for any $\bfx, \bfy \in \gf(2)^{|I|}$,}
%\begin{align}\label{eq:propFMap2}
%\sout{f (\bfx) \neq f(\bfy),}
%\end{align}
%\sout{unless $\bfx_I = \bfy_I$.} 

\begin{claim}\label{cl:fslmap} There exists a map $f : \mathbb{F}_2^{|I|} \to \mathbb{F}_Q$ satisfying Property~\ref{prpr:f} where $Q \geq 2^{t|I|}$ and $Q$ has characteristic two. \end{claim}

%Once we know the set $\Cen$, we use a secondary map $f^{(2)}_{S,\ell} : \gf(Q_1) \to \gf(Q_2)$. \textcolor{blue}{The purpose of this map is ...} Suppose as before that $S \subseteq \cS^A$, $S' \subseteq S$ where $|S'| \leq 2t$. Then, for any subset $S'' \subseteq S'$:
%\begin{align}\label{eq:propFMap2}
%\sum_{\bfx \in S''} f^{(2)}_{S,\ell}( f_{S,\ell}^{(1)}(\bfx) ) \neq 0.
%\end{align}
%
%\begin{claim}\label{cl:fslmap2} There exists a map $f^{(2)}_{S, \ell}$ where $Q_2$ is equal to the smallest prime greater than $\left(n^{\ell} \cdot s_{A} \right)^{t}$. The map $f^{(2)}_{S,\ell}$ can be uniquely described using $n + \log n$ bits. \end{claim}
%
%We assume that $\gf(Q_2) \subseteq \gf(2^n)$. 

We can apply the previous claim to determine the size of $(\bfw_1^A, \bfw_2^A)$.

\begin{claim}\label{cl:size} Assuming $|I| < \log n$, $\mathbb{F}_Q \subseteq \mathbb{F}_{2^{\bar{n}}}$, $R \leq 2^n$, and $N < 2^{\bar{n}}$, $(\bfw_1, \bfw_2)$ requires approximately
	$$ t^2 n + 2th (\ell + t) \log n $$
	bits of information. \end{claim}
\begin{IEEEproof} From the encoding, we have that the information $(\bfw_1^A, \bfw_2^A)$ is transmitted. The vector $\bfw_1^A$ has dimension $r$ which is equal to the dimension of $H_C$, a parity check matrix for an $[N, 2th + 1]_Q$ code. Approximating $r=2th (\log N + \log Q)$, $N=n^\ell$ and $Q = 2^{t |I| }$ gives that $\bfw_1^A$ requires $2th (\ell + t) \log n$ bits of information. Since $R \leq 2^n$, $\bfw_2^A$ requires approximately $t^2 n$ bits of information which gives the statement in the claim.
\end{IEEEproof}
%\tr{Should $\ell$ also appear in the formula? Comparison with our previous result, would be good too.}%\begin{IEEEproof} \textcolor{blue}{Should just be the sum of the parts $\bfw_1, \bfw_2$.}
%\end{IEEEproof}
%\Cref{cl:size}

Recall that if the method from \cite{MT02} were used roughly $thn$ bits of information exchange would be required so that the method described here requires less information exchange when $t \ll h$ and $n$ is large enough.

%The second notation is a map $O : [N] \times [N] \to \mathbb{Z}$ that operates as follows. For a vector $\bfx \in \gf(2)^n$, let $\cB_t(\bfx) = \{ \bfy : \bfy = \bfx + \bfe, \bfe \in E_t \}.$ Then denote
%$$ N(d,t) = | \cB_t(\bfx) \cap \cB_t(\bfy) |, $$
%where $d_H(\bfx, \bfy) = d$. \textcolor{blue}{again need to give the purpose of this map here. } The map $O$ is defined so that
%$$ O(i,j) = N( wt(E(i,j)), \ell). $$
%The map $O(i,j)$ can be thought of as the overlap of the Hamming error spheres for a two vectors whose difference is given by $E(i,j)$. The following lemma is well known, see Levenshtein for instance.
%
%\begin{lemma} For integers $d,t$ where $t \geq d$, 
%$$ N(d,t) = \sum_{i=0}^{t-\lceil \frac{d}{2} \rceil} \nchoosek{n-d}{i} \sum_{k=d+i-t}^{t-i} \nchoosek{d}{k}.$$
%\end{lemma} 
%
%%Let $LI(\gf(2)^n)$ be a set of elements from $\gf(2)^n$ that are linearly independent over $\mathbb{F}_2$. Clearly $|LI(\gf(2)^n)| \geq n$.  For the third stage of the encoding/decoding, we will make use of another perfect hash function $g : f(E_\ell) \to LI(\gf(2)^n)$ whose domain is the range of $f$ under the inputs from the set $E_\ell$. Using the same ideas as before, we have the following claim.

We let the matrix $H_F \in \mathbb{F}_R^{t \times N}$ from (\ref{eq:HFmatrix}) be a parity check matrix of the following form:
\begin{align*}
H_F = \begin{bmatrix}
\gamma_1                   & \gamma_2        			& \cdots & \gamma_N \\
\gamma_1^2               & \gamma_2^2    			& \cdots & \gamma_N^2 \\ 
\gamma_1^4               & \gamma_2^4   			& \cdots & \gamma_N^4 \\
\vdots  		                &                         			& \ddots & \\
\gamma_1^{2^{t-1}}  & \gamma_1^{2^{t-1}}	& \cdots & \gamma_N^{2^{t-1}}
\end{bmatrix},\vspace{-.5ex}
\end{align*}
where any subset of elements from $\{ \gamma_1, \gamma_2, \ldots, \gamma_N \}$ of size $(2^t-1) h$ is linearly independent over $\mathbb{F}_2$. Hence, we have that $\gamma_i \in \mathbb{F}_R$ where $R \geq N^{(2^t-1)h}$ (this statement follows using ideas similar to Claim~\ref{cl:mmap}). Recall $\bar{n} = n - |I|$. We assume $\mathbb{F}_R \subseteq \mathbb{F}_{2^{\bar{n}}}$. For a vector $\bfv\in \mathbb{F}_{2^{\bar{n}}}^N$ and an element $\sigma \in \mathbb{F}_{2^{\bar{n}}}$, let $\sigma(\bfv) \subseteq [N]$ return the set of positions in $\bfv$ that have value $\sigma$. Furthermore, recall that $\rk(\bfv)$ denotes the rank of $\bfv$ over $\mathbb{F}_2$ if $\bfv$ is interpreted as a $\bar{n} \times N$ matrix over $\mathbb{F}_2$. The next lemma can be used to show Property~\ref{prpr:solvesys} holds.

\begin{lemma}\label{lem:gabi} Suppose $\bfx \in \mathbb{F}_{2^{\bar{n}}}^N$ is such that $\rk(\bfx) \leq t$. If, for every non-zero value $\sigma \in  \bfx$, $|\sigma(\bfx)| \leq h$,  then $ H_F \cdot \bfx \neq 0. $
\end{lemma} {The proof of the lemma is given in an appendix.}

\subsection{Decoding}\label{ssc:t>1Dec}
In this section, we present the decoding algorithm.
Let $\cD_{\comp}$ be a decoder for the code with the parity check matrix $H_{\comp} \in \mathbb{F}_Q^{r \times N}$ from (\ref{eq:HCr}) so that for any vector $\bfv \in \mathbb{F}_Q^N$ where $\wt(\bfv) \leq th$, $\cD_{\comp}(H_{\comp} \cdot \bfv) = \bfv$. Suppose $\bfw^B = (\bfw^B_1, \bfw^B_2)$ is the result of performing steps 1) and 2) in the encoding section using the set $\cS^B$ (rather than $\cS^A$) where $\bfw^B_2 = ( \bfw^{B,(0)}_{2}, \ldots, \bfw^{B,(t-1)}_{2})$. We will also make use of a map $F : \mathbb{F}_2^{|I|} \times \mathbb{F}_{2^{\bar{n}}} \to \mathbb{F}_2^n$ that outputs a length $n$ binary vector $\bfx$ where $\bfx_I$ is equal to the first argument and $\bfx_{\minusI}$ is equal to the second argument. In the algorithm below, $\widehat{\Cen}$ contains the image of a center set under the map $f$ from (\ref{prpr:f}). For $i \in [N]$, the sets $D_i$ and $G_i$ contain elements in the same difference block. We now detail how to recover $\cS^A \bigtriangleup \cS^B$ given $\bfw^A, \bfw^B$.

\begin{enumerate}
\item Let $\dot{\bfz} = \cD_{\comp}( \bfw_{1}^{A} + \bfw_{1}^{B})=(\dot{z}_{1}, \ldots, \dot{z}_{N}) \in \mathbb{F}_Q^N$.
\item For $i \in [N]$, perform the following procedure to generate the sets $D_1, \ldots, D_N \subseteq \mathbb{F}_Q$:
\begin{enumerate}
\item If $\dot{z}_{i} = 0$, then set $D_i = \emptyset$.
\item Otherwise if $\dot{z}_i = \zeta \in \mathbb{F}_Q$, then set $D_i=\{\zeta_1, \zeta_2, \ldots, \zeta_T \} \subseteq \mathbb{F}_Q$ where $\sum_{j=1}^T \zeta_i = \zeta$, $T \leq t$.
\end{enumerate}
\item Copy $D_1, \ldots, D_N$ to $G_1, \ldots, G_N$ so that $G_1=D_1, \ldots, G_N = D_N$.
\item For $k \in \{0,1,\ldots,t-1\}$, set ${\ddot\bfw}^{(k)} = \bfw^{A,(k)}_{2} + {\bfw^{B,(k)}_{2}}.$
%\item Set $F_1 = D^{(1)}_1$, $F_2 = D^{(1)}_2$, $\ldots$, $F_N = D^{(1)}_N$.
\item From $D_1,\ldots,D_N$ update ${\ddot\bfw}^{(0)},$ $\ldots,$ ${\ddot\bfw}^{(t-1)}$ as follows:
\begin{enumerate}
\item Initialize $i=0$ and $\widehat{\Cen}=\emptyset$.
\item Set $i\leftarrow i+1$. If $i > N$ go to step 6).
\item If $\exists D_j$ where $|D_i \cap D_j| \neq 0$, do the following:
\begin{enumerate}
\item Let $\sigma \in D_i \cap D_j$.
\item For $k \in \{0,1,\ldots,t-1\}$, update ${\ddot\bfw}^{(k)} \leftarrow {\ddot\bfw}^{(k)} + H_F \cdot \bfu^{(k)}$ where $\bfu^{(k)}$ is zero except in position $j$ where $u^{(k)}_j =  E(i,j)_{\minusI} \cdot \sigma^{2^k}$.
\item Remove $\sigma$ from $D_j$. 
\item Add $(\sigma, i)$ to $\widehat{\Cen}$. Repeat step 5c).
\end{enumerate}
\item If $\nexists D_j$ where $|D_i \cap D_j| > 0$, go to step 5b). 
\end{enumerate}
\item For $k \in \{0,1,\ldots, t-1\}$, do the following: from ${\ddot\bfw}^{(k)}$, compute $\ddot{\bfz}^{(k)}$ such that ${\ddot\bfw}^{(k)} = H_F \cdot \ddot{\bfz}^{(k)}$ where the locations of the non-zero entries in $\ddot{\bfz}^{(k)}$ are equal to the locations of the non-zero entries in $\dot{\bfz}$, and $\rk(\ddot{\bfz}^{(k)}) \leq t$.
\item Initialize $\cF = \emptyset$.
\item Add the center set to $\cF$ by setting $i=0$, and doing the following:
\begin{enumerate}
\item Set $i \leftarrow i+1$. If $i > N$, then exit.
\item If $G_i = \emptyset$, then go to step 8a).
\item Suppose $G_i = \{ \sigma_1, \sigma_2, \ldots, \sigma_T \}$ where $T \leq t$. Then let 
\[H_i = \begin{bmatrix} \sigma_1 & \sigma_2 & \ldots & \sigma_T \\ 
\sigma_1^2 & \sigma_2^2 & \ldots & \sigma_T^2 \\
\sigma_1^{2^{t-1}} & \sigma_2^{2^{t-1}} & \ldots & \sigma_T^{2^{t-1}} \\
\end{bmatrix}.\]
\item Define the $t \times 1$ vector $\bfv=(v_1, \ldots, v_t)$ so that $v_k = \ddot{z}_{i}^{(k)}$. Let $V = (V_1, \ldots, V_T)= H_i^{-1} \cdot \bfv \in \mathbb{F}_{2^{\bar{n}}}^T$. 
\item For every $j \in T$ where $V_j \neq 0$, if $(\sigma_j,i) \in \widehat{\Cen}$ add $F(f^{-1}(\sigma_j),V_j)$ to $\cF$. Otherwise, if $(\sigma_j,i) \not \in \widehat{\Cen}$, let $p$ be such that $(\sigma_j, p) \in \widehat{\Cen}$. Add $F(f^{-1}(\sigma_j),V_j) + E(p,i)$ to $\cF$.
\item Go to step 8a).
\end{enumerate}
%\item Augment $\cF$ to include the additional remaining elements in the symmetric difference outside the center set as follows:
%\begin{enumerate}
%\item Initialize $i=0$.
%\item Set $i=i+1$. If $i > N$, the process is complete.
%\item If $\exists G_j$ where $|G_i \cap G_j| \neq 0$, do the following:
%\begin{enumerate}
%\item Let $\sigma \in G_i \cap G_j$.
%\item Suppose $\bfx \in \cF$ where $\bfx_I = f^{-1}(\sigma)$. 
%\item Remove $\sigma$ from $D_j$. Repeat step 5c).
%\end{enumerate}
%\item If $\nexists D_j$ where $|D_i \cap D_j| > 0$, go to step 5b). 
%\end{enumerate}
\end{enumerate}

The following theorem can be proven using the ideas introduced at the beginning of Section~\ref{sec:t>1}. A proof is included in the appendix.

\begin{theorem}\label{th:thlsets} At the end of the decoding, $\cF = \cS^A \bigtriangleup \cS^B$. 
\end{theorem}

\section{Conclusion}\label{sec:conclude}

In this work, we studied the problem of synchronizing two sets of data where the size of the symmetric difference is small and the elements in the symmetric difference are related through the Hamming metric. We provided upper and lower bounds on the minimal amount of information exchange required, in a single round of communication, to synchronize these sets. In addition, we provided transmission schemes for certain cases of this problem that require less bits of information exchange than existing algorithms. Future work involves devising improved transmission schemes and, in particular, designing schemes that work for the setup where $t>1$ without any restrictions on the elements on the symmetric difference.

\textbf{Acknowledgment:} This work was funded by the NISE program at SSC Pacific.

\begin{appendices}
\section{Proof of Theorem~\ref{th:asmptbnds}}
For an integer $q$, let $H_{q}\left(x\right)=-x\log_{q}x-(1-x)\log_{q}\left(1-x\right)+x\log_{q}\left(q-1\right).$ Furthermore, let $V_q(n,k)$ denote the size of the Hamming sphere of radius $k$ in $\mathbb{F}_q^n$. We will find the following inequalities and lemma useful. From~\cite{R06}, for $0\le k/n\le1-1/q$, we have
\begin{align}
\frac{1}{n+1}q^{nH_{q}\left(k/n\right)}\le V_{q}\left(n,k\right) & \le q^{nH_{q}\left(k/n\right)}\label{eq:sphere}
\end{align}

\begin{lemma}
	\label{lem:entropy}If $\kappa_{n}=\frac{k}{n}$ is between
	0 and $1-1/q$ and bounded away from both, and $\epsilon_{n}=o(1)$, then 
	\[
	\log_{q}V_{q}\left(n,n(\kappa_{n}+\epsilon_{n})\right)=n\left(H_{q}\left(\kappa_{n}\right)+o\left(1\right)\right).
	\]
\end{lemma}
\begin{proof}
	For $0<x_n<1$, where $x_n$ is bounded away from 0 and 1, we have 
	$
	H_{q}\left(x_n+\epsilon_{n}\right)=H_{q}\left(x_n\right)+o(1).
	$	
	Hence,
	\begin{align*}
	\log_{q}V_{q}\left(n,n(\kappa_{n}+\epsilon_{n})\right) & =nH_{q}\left(\kappa_{n}+\epsilon_{n}\right)+o(n)\\
	& =nH_{q}\left(\kappa_{n}\right)+o(n),
	\end{align*}
	where the first equality follows from (\ref{eq:sphere}). 
\end{proof}
The following simple inequalities will also be of use
%\begin{align}
%\binom{n}{k} & =\frac{n\left(n-1\right)\cdots\left(n-k+1\right)}{k\left(k-1\right)\dotsm\left(1\right)}\ge\left(\frac{n}{k}\right)^{k}\label{eq:factLB},\\
%\binom{n}{k} & \le\frac{n^{k}}{k!}\le\left(\frac{ne}{k}\right)^{k},\label{eq:factUB}
%\end{align}
\begin{align}
%\binom{n}{k} & =\frac{n\left(n-1\right)\cdots\left(n-k+1\right)}{k\left(k-1\right)\dotsm\left(1\right)}\ge\left(\frac{n}{k}\right)^{k}\label{eq:factLB},\\
\left(\frac{n}{k}\right)^{k}\le\binom{n}{k} & \le\frac{n^{k}}{k!}\le\left(\frac{ne}{k}\right)^{k},\label{eq:fact}
\end{align}
where the last inequality follows from Stirling's approximation. 
%and from (\cite{MS77}, p.309), it follows that
%\begin{equation}\label{eq:ross-2}
%\binom{n}{t}\ge\frac{1}{n}2^{nH\left(t/n\right)}.
%\end{equation}
%
%\begin{lemma}\label{lem:entropy}
%	For $0\le x\le\frac{1}{2}$, we have $ \lb H\left(x\right)\le H\left(x\right)\le\ub H\left(x\right),$
%	where $\lb H\left(x\right)=x\lg\frac{2}{x}$ and $\ub H\left(x\right)=x\lg\frac{e}{x}$,
%	and where, as a convention, $0\lg0=0$.\end{lemma}
%\begin{proof}
%The lower bound on $H$ is proved by noting that $\frac{d^{2}}{\left(dx\right)^{2}}\left(H\left(x\right)-\lb H\left(x\right)\right)=\frac{-1}{\left(1-x\right)\ln2}<0$
%for $0<x<\frac{1}{2}$ and that $H\left(0\right)-\lb H\left(0\right)=H\left(\frac{1}{2}\right)-\lb H\left(\frac{1}{2}\right)=0$.
%The upper bound follows from the facts that $\frac{d}{dx}\left(\ub H\left(x\right)-H\left(x\right)\right)=\lg\frac{1}{1-x}>0$
%for $0<x<\frac{1}{2}$ and $\ub H\left(0\right)-H\left(0\right)=0$.
%\end{proof}

We now turn to the proof of Theorem~\ref{th:asmptbnds}. Recalling the definitions of $r_1$, $r_2$, and $|\cC|$ from Theorem~\ref{th:bulbH}, we first find bounds for the terms appearing in that theorem. From Lemma \ref{lem:entropy}, we find
\begin{align}
\log_{q}r_{1} & =\log_{q}V_{q}\left(n,\lfloor\ell/2\rfloor\right)=nH_{q}\left(\frac{\lambda}{2}\right)+o\left(n\right),\label{eq:r1}\\
\log_{q}r_{2} & =\log_{q}V_{q}\left(n,\ell\right)=nH_{q}\left(\lambda\right)+o\left(n\right),\label{eq:r2}
\end{align}
{where $\lambda = \frac{\ell}{n}$. }
Using (\ref{eq:fact}) and (\ref{eq:r1}),
\begin{align*}
\log_{q}\binom{r_{1}}{h-1} & \ge\left(h-1\right)\log_{q}\frac{r_{1}}{h-1}\\
& =\left(h-1\right)nH_{q}\left(\frac{\lambda}{2}\right)-\\
&\qquad\left(h-1\right)\log_{q}\left(h-1\right)+o\left(hn\right)\\
& =hnH_{q}\left(\frac{\lambda}{2}\right)-h\log_{q}h+o\left(hn\right).
\end{align*}
Using (\ref{eq:fact}) and (\ref{eq:r2}),
\begin{align*}
\log_{q}\binom{r_{2}}{h} & \le h\log_{q}\frac{r_{2}e}{h}\\
& =h\left(nH_{q}\left(\lambda\right)+o\left(n\right)\right)-h\log_{q}h+O\left(h\right)\\
& =hnH_{q}\left(\lambda\right)-h\log_{q}h+o\left(hn\right).
\end{align*}
Finally, $\log_{q}|\cC|=O\left(n\right)$ and
so 
\begin{equation}
\log_{q}\binom{|\cC|}{t}=O\left(nt\right).\label{eq:c}
\end{equation}

From Theorem~\ref{th:bulbH}, (\ref{eq:r1}), and (\ref{eq:c}), it follows that
\begin{align*}
\log_{q}\chi\left(\mathcal{G}^{2}\right) & \ge\log_{q}\left(\binom{|\cC|}{t}\binom{r_{1}}{h-1}^{t}\right)\\
& \ge tnhH_{q}\left(\frac{\lambda}{2}\right)-th\log_{q}h+o\left(thn\right),
\end{align*}
and so 
$
\frac{\log_{q}\chi\left(\mathcal{G}^{2}\right)}{tnh}\ge H_{q}\left(\frac{\lambda}{2}\right)-\eta+o(1).
$
In our derivation of the upper bound, we need $h<r_{2}/2$. To ensure
this condition, we assume $\eta$ is less than, and bounded away from,
$H_{q}\left(\lambda\right)$. From Theorem~\ref{th:bulbH}, 
$
\chi\left(\mathcal{G}^{2}\right)\le2tq^{2tn}h^{2t}\binom{r_{2}}{h}^{2t}
$
and using (\ref{eq:r1}), and (\ref{eq:c}), 
$
\frac{\log_{q}\chi\left(\mathcal{G}^{2}\right)}{tnh}\le2(H_{q}\left(\lambda\right)-\eta)+o(1).
$

\section{Proof of Lemma~\ref{lem:gabi}}
{For a vector $\bfv$, let $dis(\bfv)$ denote the set of non-zero elements in $\bfv$ (the ``distinct'' elements in $\bfv$).} Suppose $x_{i_1}, x_{i_2}, \ldots, x_{i_{|dis(\bfx)|}}$ are elements in $\bfx$ that have distinct values. Then, we can write:
\begin{align*}
H_F \cdot \bfx &=  \sum_{j=1}^{|dis(\bfx)|} x_{i_j} \cdot \left( \sum_{k \in x_{i_j}(\bfx)} \begin{bmatrix}
\gamma_{i_{k}} \\
\gamma_{i_{k}}^2 \\
\vdots \\
\gamma_{i_{k}}^{2^{t-1}}.
\end{bmatrix} \right) \\
&= \sum_{j=1}^{|dis(\bfx)|} x_{i_j} \begin{bmatrix}
\sum_{k \in x_{i_j}(\bfx)} \gamma_{i_{k}} \\
(\sum_{k \in x_{i_j}(\bfx)} \gamma_{i_{k}})^2 \\
\vdots \\
( \sum_{k \in x_{i_j}(\bfx)} \gamma_{i_{k}})^{2^{t-1}}
\end{bmatrix}.
\end{align*}
Because $\rk(\bfx) \leq t$, clearly $dis(\bfx) \leq 2^t-1$. Since $|x_{i_j}(\bfx)| \leq h$ and any collection of $(2^t-1)h$ elements from $\{ \gamma_1, \ldots, \gamma_N \}$ are linearly independent over $\mathbb{F}_2$, it follows $\sum_{k \in x_{i_j}(\bfx)} \gamma_{i_{k}} \neq 0$. Using similar reasoning we have that the elements 
$$ \Big\{ \sum_{k \in x_{i_1}(\bfx)} \gamma_{i_k}, \sum_{k \in x_{i_2}(\bfx)} \gamma_{i_k}, \ldots, \sum_{k \in x_{i_{|dis(\bfx)|}}(\bfx)} \gamma_{i_k}  \Big\} $$
are also linearly independent over $\mathbb{F}_2$. Let $H$ be the $t \times |dis(\bfx)|$ matrix 
\begin{align*}
H=\begin{bmatrix} 
\sum_{k \in x_{i_1}(\bfx)} \gamma_{i_k} & \cdots & \sum_{k \in x_{i_{|dis(\bfx)|}}(\bfx)} \gamma_{i_k} \\
(\sum_{k \in x_{i_1}(\bfx)} \gamma_{i_k})^2 & \cdots & (\sum_{k \in x_{i_{|dis(\bfx)|}}(\bfx)} \gamma_{i_k})^2 \\
\vdots & \ddots & \vdots \\
(\sum_{k \in x_{i_1}(\bfx)} \gamma_{i_k})^{2^{t-1}} & \cdots & (\sum_{k \in x_{i_{|dis(\bfx)|}}(\bfx)} \gamma_{i_k})^{2^{t-1}}
\end{bmatrix}. 
\end{align*}
Let $\bfx' = (x_{i_1}, x_{i_2}, \ldots, x_{i_{|dis(\bfx)|}})$. Note that if $H_F \cdot \bfx = \bf0$, then $H \cdot \bfx' = \bf0$. where clearly $\rk(\bfx') = \rk(\bfx)$. However, if $H \cdot \bfx' = \bf0$,  $\rk(\bfx) = \rk(\bfx') \geq t+1$ from \cite{Gab85}, which is a contradiction.

\section{Proof of Theorem~\ref{th:thlsets}}
Suppose $\cS^A \bigtriangleup \cS^B = \cB_1 \cup \cB_2 \cup \ldots \cup \cB_T,$ where $T \leq t$. Since $(\cS^A, \cS^B)$ are $(t, h, \ell)$-sets, we can write 
\begin{align*}
\cB_w=\{ \bfx_w, \bfx_w +\bfe_{w,2}, \ldots, \bfx_w + \bfe_{w, T_w}  \},
\end{align*}
where $T_w \leq h$, and $wt(\bfe_{w,m}) \leq \ell$ for $m \in \{2,3,\ldots,T_w\}$. As a result of Properties~\ref{prpr:M} and \ref{prpr:f} for the maps $M$ and $f$ respectively, $D_i$ (for $i \in [N]$) at step 3) of the decoding is such that
$$ D_i = \{ f(\bfx) : M(\bfx) = i \}, $$
and for any $\bfx, \bfy \in \cB_j$, $d_H(\bfx, \bfy) \geq 2\ell+1$.

Suppose that at step 5-c), $\sigma \in D_i \cap D_j$ and that $f(\bfx_I)= f(\bfy_I) = \sigma$ where $M(\bfx) = i$ and $M(\bfy)=j$. Notice that under this setup $\bfx$ and $\bfy$ belong to the same difference block so that $d_H(\bfx,\bfy) \leq \ell$. Suppose further that $\ddot{\bfw}^{(k)} = H_F \cdot \ddot{\bfz}^{(k)}$ at step 5-c-i), so that $\bfx_{\bar{I}}$ appears in $\ddot{\bfz}^{(k)}$ in position $i$ and $\bfy_{\bar{I}}$ appears in $\ddot{\bfz}^{(k)}$ in position $j$. In particular, suppose $A \subseteq \mathbb{F}_{2^{\bar{n}}}$ and that
$$ \ddot{z}^{(k)}_{j} = \sum_{\bfu \in A} \left( f(\bfu_I)\right)^{2^k} \cdot \bfu_{\bar{I}}, $$
where $\bfy \in A$ for $k \in \{0,1,\ldots,t-1\}$. At the completion of step 5-c-iv), $\ddot{\bfz}^{(k)}$ is such that $ \ddot{z}^{(k)}_{j} = \sum_{\bfu \in A'} \left( f(\bfu_I)\right)^{2^k} \cdot \bfu_{\bar{I}}$ where $A'$ is the result starting with $A$ and then replacing the element $\bfy$ with $\bfx$. To see this, notice that since $f(\bfx_I)=f(\bfy_I)$, from condition 4) of $(t,h,\ell)$-sets, $\bfx$ and $\bfy$ belong to the same difference block. Using Property~\ref{prpr:E}, we know that $\bfx + \bfy = E(i,j)$. Thus, at step 5-c-ii), 
$$ \ddot{\bfw}^{(k)} =  H_F \cdot  \left( \ddot{\bfz}^{(k)} + \bfu^{(k)} \right),$$
and so the $j$-th component of $\ddot{\bfz}^{(k)} + \bfu^{(k)}$, which we denote below as $z'_j$, is such that

\begin{align*}
z'_j =&  \sum_{\bfu \in A} \left( f(\bfu_I)\right)^{2^k} \cdot \bfu_{\bar{I}} + E(i,j)_{\bar{I}} \cdot \sigma^{2^k}\\
=& \sum_{\bfu \in A} \left( f(\bfu_I)\right)^{2^k} \cdot \bfu_{\bar{I}} + (\bfx + \bfy)_{\bar{I}} \cdot f(\bfx_I)^{2^k}\\
=& \sum_{\bfu \in A'} \left( f(\bfu_I)\right)^{2^k} \cdot \bfu_{\bar{I}},\\
\end{align*}
as desired. Since this process is repeated at step 5-c) for every $i,j$ where $\exists \sigma \in D_i \cap D_j$, the vector $\ddot{\bfz}^{(k)}$ will contain only contain linear combinations of at most $t$ elements (one element per difference block) from $\cS^A \bigtriangleup \cS^B$ at step 6). In addition, since there are at most $th$ elements in the symmetric difference, the number of non-zero coefficients in $\ddot{\bfz}^{(k)}$ is at most $th$ for $k \in \{0,1,\ldots, t-1\}$ at step 6). Thus, invoking Property~\ref{prpr:solvesys}, we can recover $\{ \ddot{\bfz}^{(0)}, \ddot{\bfz}^{(1)}, \ldots, \ddot{\bfz}^{(t-1)} \}$.

Suppose $C \subseteq \mathbb{F}_{2^{\bar{n}}}$ where, from the discussion in the previous paragraph, $|C| \leq t$. For any $j \in [N]$, at step 8), we can write:
$$ \ddot{z}_j^{(k)} =  \sum_{\bfu \in C} \left( f(\bfu_I)\right)^{2^k} \cdot \bfu_{\bar{I}}. $$
Recall that from step 3) the set $\{ f(\bfu_I) : \bfu \in C \}$ is known. Furthermore, from Property~\ref{prpr:f} of the map $f$, the elements $\{ f(\bfu_I) : \bfu \in C \}$ are linearly independent over $\mathbb{F}_2$ and so the matrix $H_i$ at step {8}c) has full rank. Notice then that at step 8e), we can recover one element from each of the difference blocks from $F(f^{-1}(\sigma_j),V_j)$ and using the function $E$ it is possible to recover the remaining elements in each difference block using ideas similar to those used in Theorem~\ref{th:1sets}.

\end{appendices}


\begin{thebibliography}{10}
\providecommand{\url}[1]{#1}
\csname url@rmstyle\endcsname
\providecommand{\newblock}{\relax}
\providecommand{\bibinfo}[2]{#2}
\providecommand\BIBentrySTDinterwordspacing{\spaceskip=0pt\relax}
\providecommand\BIBentryALTinterwordstretchfactor{4}
\providecommand\BIBentryALTinterwordspacing{\spaceskip=\fontdimen2\font plus
\BIBentryALTinterwordstretchfactor\fontdimen3\font minus
  \fontdimen4\font\relax}
\providecommand\BIBforeignlanguage[2]{{%
\expandafter\ifx\csname l@#1\endcsname\relax
\typeout{** WARNING: IEEEtran.bst: No hyphenation pattern has been}%
\typeout{** loaded for the language `#1'. Using the pattern for}%
\typeout{** the default language instead.}%
\else
\language=\csname l@#1\endcsname
\fi
#2}}


\bibitem{MS77}
{F.J. MacWilliams and N.J.A. Sloane,} \textit{The theory of error-correcting codes}, North Holland Publishing Company, 1977.

\bibitem{FKS84}
{M. Fredman, M. Komlos, E. Szemeredi},'' Storing a sparse table with $O(1)$ worst case access time,'' \textit{Journal of ACM}. vol. 31, no. 3, pp. 538-544, 1984. 

\bibitem{Gab85}
{\`E.~M.~Gabidulin}, ``Theory of codes with maximum rank distance,'' \textit{Probl. Peredachi Inf.} vol. 21, no. 1, pp. 3-16, 1985.

\bibitem{L90}
{R. J. Lipton}, ``Efficient checking of computations,'' \textit{STACS}, 1990.

\bibitem{W01}
{D. B. West}, \textit{Introduction to graph theory}, Prentice Hall Upper Saddle River, 2001, vol. 2.

\bibitem{MT02}
{Y. Minsky and A. Trachtenberg}, ``Practical set reconciliation,''  \textit{Tech. Rep.}, Department of Electrical and Computer Engineering, Boston University, 2002.

\bibitem{KLT03}
{M. Karpovsky, L. Levitin, and A. Trachtenberg}, ``Data verification and reconciliation with generalized error-control codes,'' \textit{IEEE Trans. Info. Theory}, vol. 49, no. 7, pp. 1788-1793, July 2003.

\bibitem{MTZ03} 
{Y. Minsky, A. Trachtenberg, R. Zippel}, ``Set reconciliation with nearly optimal communication complexity,'' \textit{IEEE Trans. Inform. Theory}, vol. 49, no. 9, pp. 2213-2218, Sept. 2003.

\bibitem{R06}
{R. Roth}, \textit{Introduction to coding theory}, Cambridge University Press, 2006.

\bibitem{KBE11}
{T. Klove, B. Bose, and N. Elarief}, ``Systematic, single limited magnitude error correcting codes for flash memories,'' \textit{IEEE Trans. Info. Theory}, July 2011.

\bibitem{EGUV11}
{D. Eppstein, M. Goodrich, F. Uyeda, G. Varghese}, ``What's the difference? Efficient set reconciliation without prior context,'' \textit{SIGCOMM} 2011.

\bibitem{GM11}
{M. T. Goodrich and M. Mitzenmacher}, ``Invertible bloom lookup tables,'' ArXiv e-prints, 2011.

\bibitem{GL13}
{D. Guo and M. Li}, ``Set reconciliation via counting bloom filters,'' \textit{IEEE Trans. Knowledge and Data Eng.}, 2013.

\bibitem{SR14}
{V. Skachek, M. Rabbat}, ``Subspace synchronization: a network-coding approach to object reconciliation,'' \textit{IEEE International Symposium on Information Theory (ISIT)}, 2014.

\bibitem{conf1}
{R. Gabrys and F. Farnoud}, ``Reconciling similar sets of data,'' in 2015 \textit{IEEE International Symposium on Information Theory (ISIT)}, 2015, pp. 2837--2841.

\bibitem{conf2}
{R. Gabrys and F. Farnoud}, ``Reconciling similar sets of data,'' in \textit{55th Annual Allerton Conference on Communication, Control, and Computing}, Monticello, IL, 2017.

%\bibitem{L01}
%{V.I. Levenshtein}, ``Efficient reconstruction of sequences from their subsequences or super sequences,'' \textit{Journal of Comb. Theory}, Vol. 93, No. 2, pp. 310-332, 2001.

%\bibitem{L02}
%{V.I. Levenshtein}, ``Bounds for deletion/insertion correcting codes,'' \textit{ISIT}, 2002.


%\bibitem{CKYYZ14} {D. Chen, C. Konrad, K. Yi, W. Yu, Q. Zhang}, ``Robust set reconciliation,'' \textit{SIGMOD}, 2014.

%\bibitem[SK83]{SK83}
%{D. Sankoff and J. B. Kruskal, Eds.}, \textit{Time warps, string edits, and macromolecules: the theory and practice of sequence comparison}. Addison-Wesley Pub. Co., Advanced Book Program, 1983.

%\bibitem{S}
%{N.J.A. Sloane}, ``On single deletion correcting codes,'' available online at http://www.research.att.com/\textasciitilde njas.


\end{thebibliography}
\end{document}